\documentclass[12pt]{article}
\usepackage[]{graphicx}\usepackage[]{color}
\makeatletter
\def\maxwidth{ %
  \ifdim\Gin@nat@width>\linewidth
    \linewidth
  \else
    \Gin@nat@width
  \fi
}
\makeatother

\usepackage{framed}
\makeatletter
 {\par\unskip\endMakeFramed%
 \at@end@of@kframe}
\makeatother

\usepackage{alltt}

\usepackage{amsmath}
\usepackage{amsthm}
\usepackage{graphicx,psfrag,epsf}
\usepackage{amssymb}
\usepackage{natbib}
\usepackage[ruled, vlined]{algorithm2e}
\usepackage{chngcntr}

\newtheorem{theorem}{Theorem}
\newtheorem{lemma}{Lemma}
\DeclareMathOperator{\tr}{tr}
\DeclareMathOperator{\diag}{diag}
\DeclareMathOperator{\sign}{sign}

\DeclareMathOperator{\Unif}{Unif}
\DeclareMathOperator{\snr}{snr}
\DeclareMathOperator{\Var}{Var}
\DeclareMathOperator{\E}{\mathbb{E}}
\DeclareMathOperator{\TN}{TN}
\DeclareMathOperator{\CI}{CI}
\DeclareMathOperator{\bA}{\boldsymbol{A}}
\DeclareMathOperator{\bX}{\boldsymbol{X}}

\DeclareMathOperator{\bI}{\boldsymbol{I}}
\DeclareMathOperator{\bSigma}{\boldsymbol{\Sigma}}

\newtheorem{corollary}{Corollary}

\theoremstyle{definition}
\newtheorem{example}{Example}
\newtheorem{remark}{Remark}

\IfFileExists{upquote.sty}{\usepackage{upquote}}{}
\begin{document}

\def\spacingset#1{\renewcommand{\baselinestretch}%
{#1}\small\normalsize} \spacingset{1}

\title{\bf Tractable Post-Selection Maximum Likelihood Inference for the Lasso}
\author{Amit Meir\footnote{To whom correspondence should be addressed: amitmeir@uw.edu}
\\ Department of Statistics \\ University of Washington
 \and Mathias Drton \\ Department of Statistics\\ University of Washington}

\maketitle

\bigskip
\begin{abstract}
  Applying standard statistical methods after model selection may
  yield inefficient estimators and hypothesis tests that fail to
  achieve nominal type-I error rates.  The main issue is the fact that
  the post-selection distribution of the data differs from the
  original distribution.  In particular, the observed data is
  constrained to lie in a subset of the original sample space that is
  determined by the selected model. This often makes the
  post-selection likelihood of the observed data intractable and
  maximum likelihood inference difficult.  In this work, we get around
  the intractable likelihood by %%efficiently
 generating noisy unbiased
  estimates of the post-selection score function and using them in a
  stochastic ascent algorithm that yields correct post-selection
  maximum likelihood estimates.  We apply the proposed technique to
  the problem of estimating linear models selected by the lasso.  In an
  asymptotic analysis the resulting estimates are shown to be
  consistent for the selected parameters and to have a limiting truncated
  normal distribution.  Confidence intervals constructed based on the
  asymptotic distribution obtain close to nominal coverage rates in
  all simulation settings considered, and the point estimates are
  shown to be superior to the lasso estimates when the true model is
  sparse. 
\end{abstract}

\noindent%
\vfill

\newpage
\spacingset{1.45} % DON'T change the spacing!

%%% Abstract %%%%%%%%%%%%%%%%%

\section{Introduction}
\subsection{Inference After Model Selection}
Consider the linear regression model 
$$y = \bX\beta + \varepsilon,$$ 
where $y \in \mathbb{R}^{n}$ is a response vector, $\bX\in \mathbb{R}^{n\times p}$ is a matrix of covariate values and $\varepsilon \in \mathbb{R}^{n}$ is a noise vector.  When the number of available covariates $p$ is large, it is often desirable or even necessary to specify a more succinct model for the data. This is commonly done by selecting a subset of the columns of $\bX$ to serve as predictors for $y$. Here, we focus on model selection with the lasso \citep{tibshirani1996regression}, which uses an $\ell_1$ penalty to estimate a sparse coefficient vector. 

A well known, yet not as well understood problem, is the problem of performing inference after a model has been selected. In particular, it is known that  confidence intervals for parameters in selected models often do not achieve target nominal coverage rates, hypothesis tests tend to suffer from an inflated type-I error rate and point estimates are often biased. A simple Gaussian example serves well to illustrate the issues that may arise when using the same data for selection and inference.

\begin{example}\label{firstUnifExample}
  Let $Y_{1},\dots,Y_{n} \sim f$ i.i.d., with $\E_{f}(Y_i) = \mu$ and
  $\Var_{f}(Y_i) = 1$. Furthermore, suppose that estimation of $\mu$ is
  of interest only if a statistical test provides evidence that it is
  nonzero.  Specifically, suppose that at a 5\%-level, we reject
  $H_0: \mu=0$ if $|\bar{y}|>1.96/\sqrt{n}$.  In this setting, if
  $|\mu|<1.96/\sqrt{n}$, the uncorrected estimator $\hat\mu=\bar{y}$
  will overestimate the magnitude of $\mu$ whenever we choose to
  estimate it.
\end{example}

An example of early work emphasizing the fact that data-driven model
selection may invalidate standard inferential methods is the article
by \citet{cureton1950validity}, with its aptly chosen title `validity,
reliability and baloney'.  Subsequently, this problem has been studied
in the context of regression modeling. In particular, it has been
shown that it is impossible to uniformly approximate the
post-selection distribution of linear regression coefficient estimates
\citep{Potscher91, leeb2005model, leeb2006can}.  

The field of post-selection (or selective) inference is concerned with
developing statistical methods that account for model selection in
inference.  
%These methods enable researchers to perform valid
%inference in the face of data-driven model selection.  The field has
%seen much growth in recent years. 
The majority of work in selective
inference is concerned with constructing confidence intervals and
performing tests after model selection; see for example
\citet{lee14marginal}, \citet{taylor2014exact},
\citet{benjamini2005false}, \citet{Weinstein13}, and
\citet{rosenblatt2014selective}.  The particular case of model
selection with $\ell_1$ penalization is treated by
\citet{Lee16} and \citet{lockhart2014significance}.
\citet{Fithian15} consider the general problem of testing
after model selection.  Estimation after model selection is in the
focus of the work of \citet{reid2014post},
\citet{benjamini2014selective}, and \citet{Routenberg15}.

In order to reconcile the aforementioned impossibility results with
the recent advances in post-selection
inference, we must clearly define the targets of inference.

\subsection{Targets of Inference}

In the context of variable selection in regression, let
$\mathcal{M} := \mathcal{P}(\{1,\dots,p\})$ be the set of models under
consideration, defined as the power set of the indices of the columns
of the design matrix $\bX$.  Further, let
$S : \mathbb{R}^{n} \rightarrow \mathcal{M}$ be a model selection
procedure that selects a model $M \in \mathcal{M}$ based on the
observed data $y \in \mathbb{R}^{n}$.

When discussing estimation after model selection in linear regression, one may consider two different targets for inference.  The first are the `true' parameter values in correct models where all variables with non-zero coefficient are present. An alternative target
for estimation is the vector of regression coefficients in the selected model
\begin{equation}\label{naiveest} 
\beta_0(y)= (\bX_M^{T}\bX_M)^{-1}\bX^{T}_M \E(Y).
\end{equation}
In \eqref{naiveest}, $M=S(y)$ is the selected model, and $\bX_M$ is the
sub-matrix of $\bX$ made up of the columns indexed by $M$.  These two
targets of estimation coincide when the selected model is true,
meaning that it contains all variables that have a non-zero regression
coefficient.  Indeed, if the observed value $y$ is such that
$S(y) = M$ for a model $M$ that contains all covariates with non-zero
coefficients, then $\E(y) =\bX_M\beta_0^M$ and $\beta_0^M =
\beta_0(y)$. Here $\beta_0^M$ is the vector of non-zero true
coefficients padded with zeros to make it a vector of length $|M|$.

\citet{Potscher91} and \citet{leeb2003finite} study the
behavior of least squares coefficients as estimators of the true
regression coefficients in a sequential testing setting. In contrast, works such as
\citet{Berk13} and \citet{leeb2015various} consider inference
with respect to the regression coefficients in the selected model. In
this work, we follow the latter point of view, taking the stance that
a true model does not necessarily exist or, even if one exists, may be
difficult to identify. Thus, the interest is in the
parameters of the model the researchers have decided to investigate.

\subsection{Conditioning on Selection} \label{sub:condOnSelection}

A data-driven model selection procedure tends to choose models that
are especially suited for the observed data rather than the
data-generating distribution. In linear regression this would often be
in the form of inclusion of variables that are correlated with the
dependent variable only due to random variation. A promising approach
for correcting for this bias towards the observed data is to condition
on the selection of a model.

\begin{figure}[t]
\begin{center}
  \includegraphics[width=6 in]{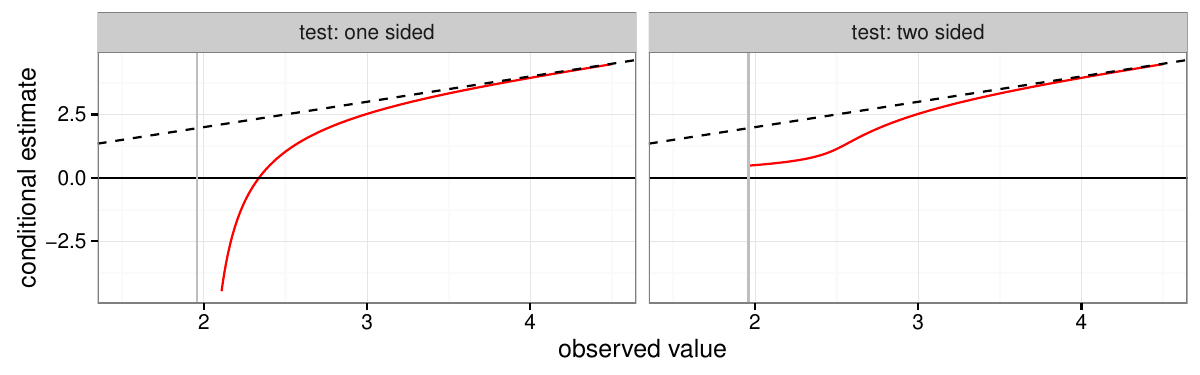} \caption[]{Conditional
    estimators for the univariate normal distribution exhibiting an adaptive shrinkage behavior. The right panel describes the conditional estimator
    when selection is two-sided: $|y| > 1.96$.  The left panel
    describes one-sided selection: $y > 1.96$. The red line plots the
    value of the conditional estimator as a function of the observed
    value, the dashed line is the $x = y$ line and the grey line marks
    the threshold. In the left plot, the conditional estimate
    asymptotes to $-\infty$ as the observed value approaches the
    threshold.}
\label{univFigure}
\end{center}
\end{figure}

\begin{example}
  Consider once again the univariate normal example, simplified via
  sufficiency to a single observation. Let $Y \sim N(\mu, 1)$ and
  assume that we are interested in estimating $\mu$ if and only if
  $|Y|>c$ for some constant $c>0$. Standard inferential techniques
  assume that we observe values from the distribution
  $Y \sim N(\mu, 1)$. However, when inference is preceded by testing
  we never observe any values $-c < Y < c$ and the post-selection
  distribution of the observed value is not normal but truncated
  normal.  Thus, the conditional post-selection maximum likelihood
  estimator (MLE) is:
$$
\hat\mu = \arg\max_\mu f(y|\{|Y|>c\}) = \arg\max_\mu
\frac{f_\mu(y)}{P(|Y|>c)} I_{\{|Y|>c\}}.
$$

The right-hand panel of Figure \ref{univFigure} plots the
post-selection MLE (as a function of $y$) for the two-sided case
described above.  Since this MLE is an even function we show the graph
only for $y>0$. The left-hand panel describes the post-selection MLE
for the one-sided case where we estimate $\mu$ if $y>c$. In the
two-sided case the estimator is an adaptive shrinkage estimator that
shrinks the observed value towards zero when it is close to the
threshold and keeps it as it is when its magnitude is far away from
the threshold.
\end{example}

More generally, let $Y \sim f_\theta$ follow a distribution from an
exponential family with sufficient statistic
$T(Y)\in\mathbb{R}^{p}$. The likelihood of $T(y)$ given that model $M$
has been selected is
$$
\mathcal{L}_M(\theta) = \frac{P(M|T(y)) f(T(y))}{P(M)} I_{M},
$$
where we use the shorthand
$P\left(M|T(y)\right) := P\left(S(Y) = M | T(Y) = T(y)\right)$ for the
conditional probability of selecting model $M$ given $T(y)$.
Similarly, $P(M) := P(S(Y) = M)$ is the unconditional probability of
selecting $M$, $f(T(y))$ is the unconditional density function of
$T(y)$, and $I_M=I_{\{S(y)=M\}}$ is the indicator function for the
selection event.

The main obstacle in performing post-selection maximum likelihood
inference is the computation of the probability of model selection
$P(M)$, which is typically a $p$ dimensional integral. Such integrals
are difficult to compute when $p$ is large, and much of the work in
the field of post-selection inference has been concerned with getting
around the computation of these integrals. For example,
\citet{Lee16} propose to condition on the signs of the
selected variables as well as some additional information contained in
the sub-space orthogonal to the quantity of interest in order to
obtain a tractable post-selection
likelihood. \citet{panigrahi2016bayesian} approximate $P(M)$ with a
barrier function.

Conditioning on information beyond the selection of the model of
interest, while having the benefit of providing tractable solutions to
the post-selection inference problem, may drastically change the form
of the likelihood. Consider once again the post-selection estimators
for the univariate normal problem (Figure \ref{univFigure}). Suppose
that we observe $y>0$. Then the right-hand panel plots the conditional
estimator for the scenario where two-sided testing is performed. 
%This estimator is essentially equal to the observed value $y$ when $y$ is
%far from the threshold and a value above zero when $y$ is close to the
%threshold. 
On the left-hand side we plot the conditional estimator for
$\mu$ as a function of $y$ when we condition on the two-sided
selection event as well as the sign of $y$. Indeed, since our observed
value is positive we condition on $\{|Y| > c, Y >0\} = \{Y >
c\}$. This second estimator is close to the observed value $y$ when
$y$ is far from the threshold but approaches negative infinity as
$y\rightarrow c$, see Appendix C for
details. Thus, even in the univariate normal case, conditioning on the
sign of $y$ in two-sided testing, may drastically alter the resulting
conditional estimator.

\subsection{Outline}
In this work, instead of working with the intractable post-selection
likelihood, we base our inference on the post-selection score function
which can be approximated efficiently even in multivariate problems. The
following lemma describes the post-selection score
function for exponential family distributions.

\begin{lemma}\label{scoreLemma}
  Suppose the observation $y$ is drawn from a distribution $f_\theta$
  that belongs to an exponential family with natural parameter
  $\theta$ and sufficient statistic $T(y)$.  If the model selection
  procedure $S(y)$ satisfies $P(S(Y) = M\,|\,T(y)) \in \{0,1\}$ for a given
  model $M$, then the conditional (post-selection) score function is given by:
\begin{align}
\frac{\partial}{\partial\theta} \log\mathcal{L}(\theta) &= T(y) - \E_\theta\left(T(Y) \middle| M\right).
\end{align}
\end{lemma}
\begin{proof}
  This result follows directly from the fact that the conditional
  distribution of an exponential family distribution is also an
  exponential family distribution as long as
  $P(M|T(y)) \in \{0,1\} $. See \citet{Fithian15} for
  details.
\end{proof}

In the specific setup we consider subsequently, the conditional
distribution of $T(Y)$ given $M$ is a multivariate truncated normal
distribution. While it is then difficult to compute $\E(T(Y)|M)$, we
are able to sample efficiently from the multivariate truncated normal
distribution using a Gibbs sampler \citep{geweke1991efficient}. The
main idea behind the method we propose is to use the samples from the
truncated multivariate normal distribution as noisy estimates of
$\E(T(Y)|M)$ and take small incremental steps in the direction of the
estimated score function, resulting in a fast stochastic gradient
ascent algorithm.  Our framework has similarities with the
contrastive divergence method of \citet{hinton2002training}.

The rest of the article is structured as follows. In Section 2 we
present the proposed inference method in detail and apply it to
selective inference on the mean vector of a multivariate normal
distribution. In Section 3 we describe how the proposed framework can
be adapted for post-selection inference in a linear regression model
that was chosen by the lasso. In Section~\ref{sec:asympt-cond-estim}
we formulate conditions under which the conditional MLE is consistent.  A
simulation study in Section 5 demonstrates that the proposed approach
yields improved point estimates for the regression coefficients, and that our confidence intervals, despite lacking a rigorous theoretical justification, achieve close to nominal coverage
rates. Finally, in Section 6 we conclude with a discussion.

\section{Inference for Selected Normal Means}\label{sec:selectedMeans}
Before considering the Lasso, we first discuss the simpler problem of
selectively estimating the means of a multivariate normal
distribution. Let $Y \sim N(\mu, \bSigma)$ with mean vector
$\mu \in \mathbb{R}^{p}$ and a \emph{known} covariance matrix
$\bSigma$.  Observing $y$, we select the model
\begin{equation}\label{mvnSelection}
M = \left\{j\in \{1,\dots,p\} :\  y_j \leq l_j\ \text{or}\ y_j  \geq u_j \right\},
\end{equation}
where $l_1,\dots,l_p,u_1,\dots,u_p \in [-\infty, \infty]$ are
predetermined constants with $l_1 < u_1,\dots,l_p < u_p$. We
then perform inference for the coordinates $\mu_j$ with
$j \in M$ (or possibly inference for a function of these coordinates).

This seemingly simple problem has garnered much attention. For the
univariate case of $p=1$, \citet{Weinstein13} propose a
method for constructing valid confidence intervals, and
\citet{benjamini2014selective} compute the post-selection MLE for
$\mu$. For $p\gg1$, \citet{Lee16} develop a recipe for
constructing valid confidence intervals for the selected means or
linear functions thereof. \citet{reid2014post} discuss ML estimation
when $\bSigma = \sigma^{2} \bI$.  To the best of our knowledge, the method we propose
below is the first to address the computation of the conditional MLE
when $p\gg1$ and the covariance matrix $\bSigma$ is of general
structure.

Conditionally on selection, the distribution of $y$ is truncated multivariate normal, as the $j$th coordinate of $y$ is constrained to lie in the interval $(l_j,u_j)$ if $j\notin M$ or in its complement if $j\in M$. 
%The conditional post-selection likelihood can be written as 
%$$
%\mathcal{L}(\mu) = \frac{\varphi(y; \mu,\Sigma)}{P_{\mu}(M)} I_{M},
%$$
%where $\varphi$ is the multivariate normal density. 
In Section \ref{sub:truncnormsamp} we describe the Gibbs sampler we use to sample from a truncated multivariate normal distribution, in Section \ref{sub:SGA}  we describe how such samples can be used to compute the post-selection estimator and in Section \ref{sub:normCI} we propose a method for constructing confidence intervals based on the conditional MLE and samples obtained from the truncated normal distribution.

\subsection{Sampling from a Truncated Normal Distribution}\label{sub:truncnormsamp}
Sampling from the truncated multivariate normal distribution is a well studied problem \citep{griffiths2004gibbs, pakman2014exact}. We
choose to use the Gibbs sampler of \citet{kotecha1999gibbs}, as it is
especially suited to our needs and simple to implement. 

Assume we wish to generate a draw from the univariate truncated normal
distribution constrained to lie in the interval
$[l,u]\subseteq[-\infty,\infty]$. This distribution has CDF
$$
\Phi(y ; \mu,\sigma^{2},l,u) := \frac{\Phi(y;\mu,\sigma^{2}) - \Phi(l;\mu,\sigma^{2})}{\Phi(u;\mu,\sigma^{2}) - \Phi(l;\mu,\sigma^{2})},
$$
where $\Phi(y;\mu,\sigma^{2})$ denotes the CDF of the (untruncated)
univariate normal distribution with mean $\mu$ and variance
$\sigma^2$. A simple method for sampling from the truncated normal
distribution samples a uniform random variable $U\sim U(0,1)$
and sets
\begin{equation}\label{univtruncatedinv}
y = \Phi^{-1}(U;\mu,\sigma^{2},l,u) = \Phi^{-1}\left(U\left(\Phi(u) - \Phi(l)\right) - \Phi(l)
;\mu,\sigma^{2}\right).
\end{equation}

Next, consider sampling from the truncated normal constrained to the
set $(-\infty, l]\cup[u,\infty)$.  In this case, we may first sample a
region within which to include $y$ and then sample from a truncated univariate normal distribution
constrained to the selected region using the formula given in
\eqref{univtruncatedinv}.

Given this preparation, we may implement a Gibbs sampler for a
truncated multivariate normal distribution as follows.  Let
$y\sim N(\mu,\bSigma)$, and let $f(y|M)$ be the conditional
distribution of $y$ given the selection event.  While the marginal
distributions of $f(y|M)$ are not truncated normal, the full
conditional distribution $f(y_j | M, y_{-j})$ for a single coordinate
$y_j$ is truncated normal with parameters
\begin{equation*}
\mu_{j,-j} = \mu_j + \bSigma_{j,-j} \bSigma_{-j,-j} (y_{-j} - \mu_{-j}), \qquad
\sigma^{2}_{j,-j} = \bSigma_{j,j} - \bSigma_{j,-j}\bSigma^{-1}_{-j,-j}\bSigma_{-j,j}.
\end{equation*}
The Gibbs sampler repeatedly iterates over all
coordinates of $y$ and draws a value for $y_j$ conditional on $M$
and $y_{-j}$.  So at the $t$th iteration we sample
$$
Y^{t}_j \sim f(y_j | M,
y_1^{t},\dots,y_{j-1}^{t},y_{j+1}^{t-1},\dots,y_{p}^{t-1}),\qquad j=1,\dots,p.
$$
The support of the truncated normal distribution is determined by
whether or not $j \in M$. 

\subsection{A Stochastic Gradient Ascent Algorithm}\label{sub:SGA}
The Gibbs sampler described above can be used to closely approximate
$\E(Y|M)$ but computation of the likelihood $\mathcal{L}_M(\mu)$
remains intractable. However, for optimization of the likelihood, we
can simply take steps of decreasing size in the direction of the
evaluated gradient
\begin{equation}\label{stochasticStep}
\mu^{i} = \mu^{i-1} +\gamma_i\bSigma^{-1} \left(y - y^{i}(\mu^{i-1})\right),
\end{equation}
where $y$ is the observed data, $y^{i}(\mu^{i-1})$ is a sample from the truncated multivariate normal distribution taken at $\mu^{i-1}$ and the step size $\gamma_i$ satisfies:
\begin{equation}\label{gammaconditions}
\sum_{i=1}^{\infty} \gamma_i = \infty, \;\;\;\;\;\;\;\;\;\;\;\;\;\;\;\;\;\;\;\;\;\;\;
\sum_{i=1}^{\infty} \gamma_i^2 < \infty.
\end{equation}
\begin{figure}[t]
\begin{center}
\includegraphics[width= 6 in]{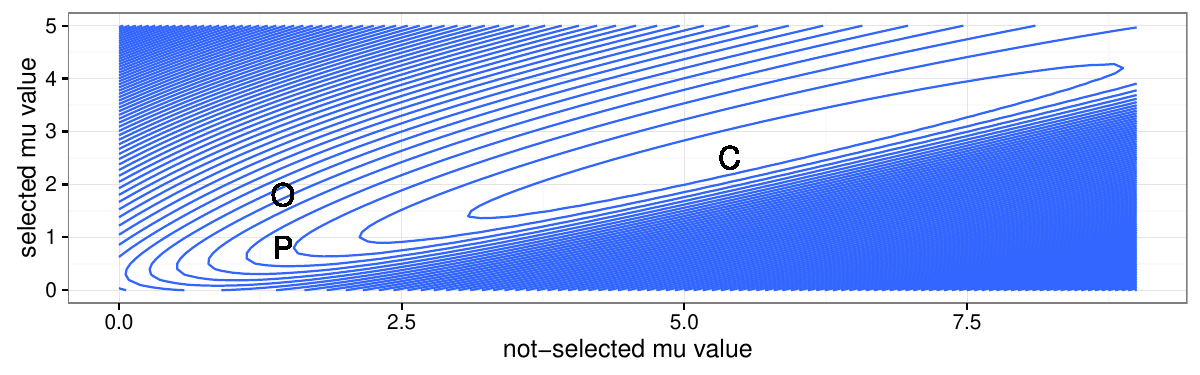} 
\caption[]{The contours of the conditional log-likelihood of a
  two-dimensional normal distribution. A selection rule $|y_j| > 1.65$
  was applied to the observed value $y = (1.45, 1.8)$ marked with
  `O'. The conditional MLE where all coordinates are estimated is
  marked with `C' at $(5.4, 2.5)$, and the plug-in conditional MLE which does not estimate the coordinates that were not selected is
  marked as `P' at $(1.45, 0.8)$. The plug-in estimator, unlike the full conditional MLE, is an adpative shrinkage estimator as in the univariate case. }
\label{2dllik}
\end{center}
\end{figure}

\begin{figure}[t]
\begin{center}
\includegraphics[width=5 in]{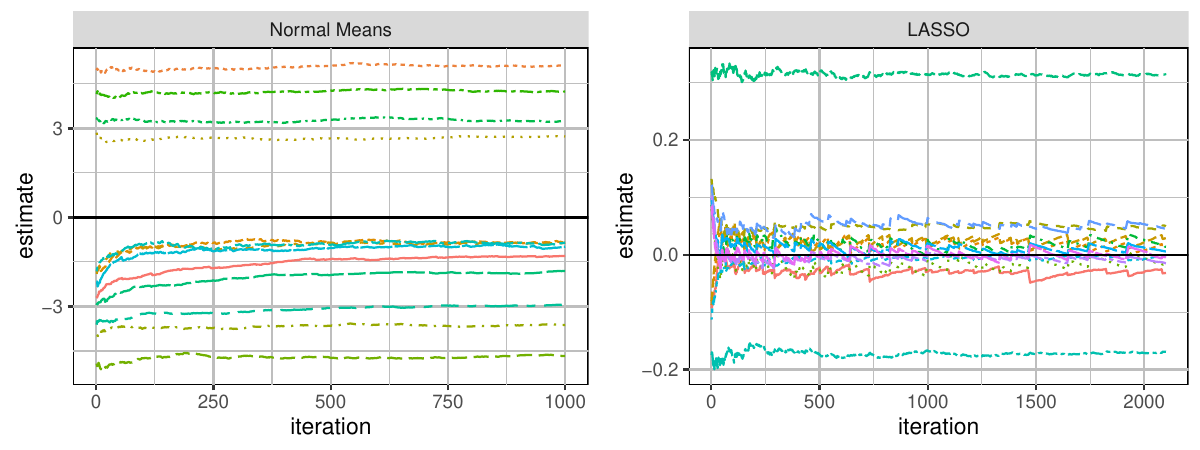} 
\caption[]{Convergence of the stochastic optimization algorithms. We
  plot the parameter estimates as a function of the number of gradient
  steps taken for the post-selection normal means estimation problem
  (left panel) and the post-selection regression estimation problem
  (right panel). The algorithms tend to converge to the neighborhood of the MLE in a few hundred iterations.}
\label{algConverge}
\end{center}
\end{figure}

We emphasize that while it is technically possible to compute an MLE for the
entire mean vector of the observed random variable, it is not
necessarily desirable. To see why, consider once again the left-hand
panel of  Figure \ref{univFigure} where the estimator tends to
$-\infty$ as the observed value approaches the threshold. Such erratic
behavior may arise when we estimate the coordinates of $\mu$ which
were not selected, based on observations that are constrained to lie
in a convex set, resulting in poor estimates
also for the selected coordinates.

\begin{example}
  \label{ex:two-dim-mean-select}
  %As an example of this phenomenon 
  We plot the conditional log-likelihood for a two-dimensional normal
  model in Figure \ref{2dllik}.  In such a low-dimensional case, the
  likelihood function can be computed using routines from the
  `mvtnorm' R package \citep{genz2016mvt}.  Our plot is for a setting
  where we observe $y = (1.45, 1.8)$ with
  $\bSigma_{ij} = 0.5^{I\{i\neq j\}}$, and only the first coordinate of
  $\mu$ was selected based on the thresholds $l_1 = l_2 = -1.65$,
  $u_1 = u_2 = 1.65$.  The point $y$ is marked in the figure as an
  `O', and the log-likelihood is maximized at the point marked with
  `C', which is $\hat\mu = (5.4, 2.5)$.  We see that instead of
  performing shrinkage on the observed selected coordinate, the
  selected coordinate was estimated to be far larger than the
  observed value.
\end{example}

In order to mitigate this behavior, we propose using a plug-in estimator for the coordinates outside of $M$. Particularly, we limit ourselves to taking steps of the form
\begin{equation}\label{mvtsteps}
\mu_j^i = 
\begin{cases}
\mu^{i-1} +\gamma_i\bSigma^{-1}_{j,.} \left(y - y^{i}(\mu^{i-1})\right)
& \text{if} \ j \in M, \\
y_j & \text{if} \ j\notin M,
\end{cases}
\end{equation}
where $\bSigma^{-1}_{j,.}$ is the $j$th row of $\bSigma^{-1}$.  In other
words, we impute the unselected coordinates of $\mu$ with the
corresponding observed values of $y$, and maximize the likelihood only
with respect to the selected coordinates of $\mu$. These plug-in
estimates for the coordinates of $\mu$ which were not selected are consistent, as
we show in Section~\ref{sec:asympt-cond-estim}.  The plug-in
conditional MLE for Example~\ref{ex:two-dim-mean-select} is shown as a
`P' in Figure \ref{2dllik}.  It is approximately
$\hat\mu = (1.45, 0.8)$.

Next, we give a convergence statement for the proposed
algorithm. Since our gradient steps are based on $y^{i}(\mu^{i-1})$, a
noisy estimate of $\E_{\mu^{i-1}}(Y|M)$, the resulting algorithm fits
into the stochastic optimization framework of
\citet{Bertsekas00}. In short, the theory for stochastic
optimization guarantees that taking steps in the form of
\eqref{stochasticStep} leads to convergence to the MLE as long as
the variance of the gradient steps can be bounded.

\begin{theorem}\label{thm:nmeans}
Let $Y \sim N(\mu,\bSigma)$, and let $M$ be defined as in \eqref{mvnSelection}. Then for all $j\in M$:
$$
E_\mu\left(Y^{i}_j(\mu) - \E_{\mu}(Y_j | M)\right)^2 \leq 
\frac{\tr(\bSigma)}{P\left(\bigcap_{j\notin M}\{l_j < y_j < u_j\}\right) \prod_{j\in M}\Phi(l_j; u_j, \sigma^{2}_{j, -j})}.
$$
The algorithm described in \eqref{mvtsteps} converges to the Z-estimator given by the root of the function
\begin{equation}\label{mvtzest}
\psi(\mu)_j = 
\begin{cases}
\bSigma_{j, .}^{-1}\left(y_j - \E_\mu(Y_j|M) \right)& \text{if} \ j\in M, \\
y_j - \mu_j &  \text{if} \ j \notin M.
\end{cases}
\end{equation}
\end{theorem}

%%%%% MVT SGD Algorithm

A precise description of the optimization algorithm is given in
Algorithm 1 in the appendix.   Figure
\ref{algConverge} shows typical optimization paths for Algorithm
1 as well as the stochastic gradient method for the Lasso described in 
Section \ref{sec:lassoMLE}.

%%%%%%%%%% Convergence Figure

\begin{figure}[h!]
\begin{center}
\includegraphics[width= 5.2 in ]{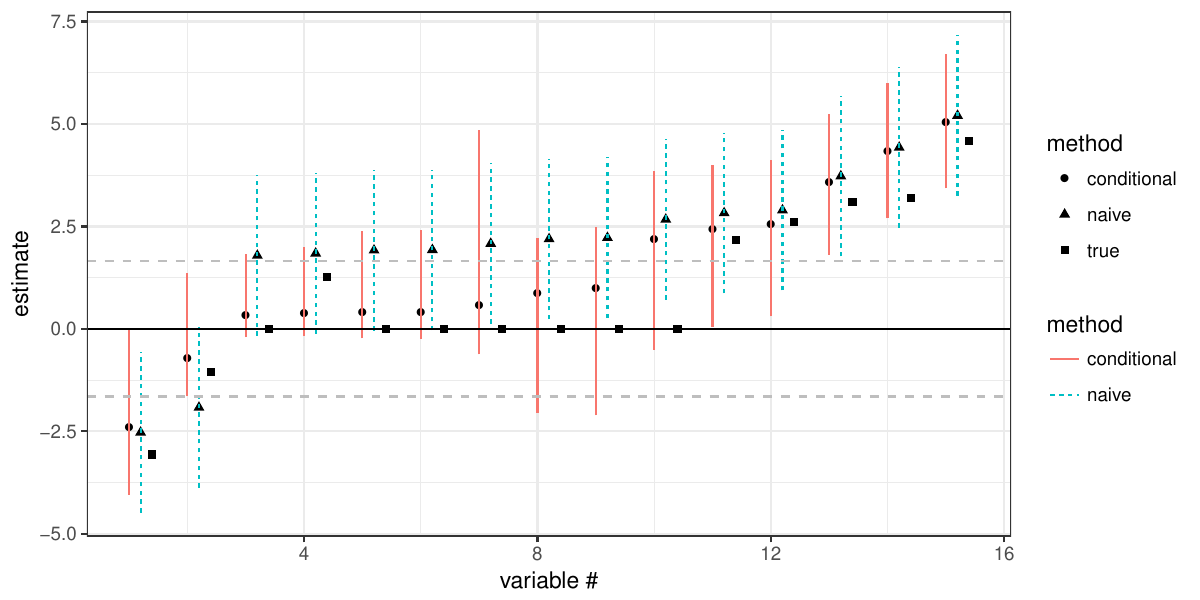} 
\caption[]{Post-selection estimates and confidence intervals for the
  normal means problem. In this example, $15$ means were selected for
  a threshold of $1.65$. The observed values are marked by triangles and the conditional estimators are marked by circles. Naive confidence intervals are marked by a dashed blue line and Conditional-Wald confidence intervals are marked by solid red lines. The true values of the parameter are shown as squares.}
\label{normalEstCI}
\end{center}
\end{figure}

\subsection{Conditional Confidence Intervals} \label{sub:normCI} 

In the absence of model selection, the MLE is typically asymptotically
normal, and it is common practice to construct Wald confidence
intervals based on this limiting distribution:
\begin{equation}\label{naiveNormalCI}
\hat{\mu}^{\text{naive}} = y, \qquad
\CI^{\text{naive}}_j = (\hat\mu^{\text{naive}}_j - z_{j,1-\alpha/2}, \hat\mu^{\text{naive}}_j - z_{j,\alpha/2}),
\end{equation}
where $z_{j,\alpha}$ denotes the $(1-\alpha)$ quantile of the
asymptotic normal distribution for the $j$th coordinate.  The
post-selection setting is more complicated, however, because we  can no longer rely the asymptotic normality of the estimators. Instead, we propose to construct confidence intervals based on the second order Taylor expansion of the conditional likelihood. 

In order to describe our proposed approximation to the distribution of the conditional
MLE, we extend the normal means problem to the setting of an
$n$-sample.  So assume that instead of observing a single vector
$y\in \mathbb{R}^{p}$, we have a set of observations
$y_1,\dots,y_n\in\mathbb{R}^p$ and perform model selection and
inference based on $\bar{y}_n = n^{-1} \sum_{i=1}^{n} y_i$. Our confidence intervals are based on the approximation
\begin{equation}\label{asympNormalDist}
\sqrt{n}(\hat{\mu}^M_n - \mu^M_0) \approx 
\sqrt{n}\text{Var}_{\mu^M_0}\left(\sqrt{n}\bSigma^{-1}\bar{Y}_n\right|M)^{-1}
\bSigma^{-1}\left(\bar{y}_n - \E_{\mu^{M}_0}(\bar{Y}_n|M)\right).
\end{equation}

Based on this approximation, we construct confidence intervals
\begin{equation}\label{normalCI}
\hat{\CI}_j = \left(\hat\mu^{M}_{j,n} - \hat\TN_{j, 1-\alpha/2} / \sqrt{n}, \; \hat\mu^{M}_{j,n} - \hat\TN_{j, \alpha/2} / \sqrt{n}\right).
\end{equation}
Here, $\hat\TN$ stands for the conditional distribution given selection
of
\begin{equation}\label{normalCIasy}
\text{Var}_{\hat\mu^M}\left(\sqrt{n}\bSigma^{-1}\bar{Y}\right|M)^{-1}
\bSigma^{-1}\sqrt{n}\left(\bar{y}_n - \E_{\hat\mu^{M}}(\bar{Y}_n|M)\right).
\end{equation}
We estimate the quantiles $\hat\TN_{j, 1-\alpha/2}$ and $\hat\TN_{j,
  \alpha/2}$
using empirical quantiles of samples from the truncated normal
distribution.  While we are unable to provide theoretical justification for these confidence intervals, a comprehensive simulation study reveals that they obtain coverage rates that are significantly better than those of the naive confidence intervals, and are surprisingly close to the desired level (Section \ref{simStudy}).

\begin{example}
  Figure \ref{normalEstCI} shows point estimates and confidence
  intervals for selected means in a normal means problems.  The figure
  was generated by sampling $Y\sim N(\mu,\Sigma)$ with
  $\mu_{1},\dots,\mu_{20}\sim N(0,4)$ i.i.d., $\mu_{21}=\dots=\mu_{100} =0$
  and $\Sigma_{i,j} = 0.3 I_{i \neq j} + 1 I_{i = j}$.  The applied
  selection rule was $S(y)=\{j: \; |y_j| > 1.65\}$.  The plotted
  estimates are the conditional estimates computed using the algorithm defined by \eqref{mvtsteps} along with the $95\%$ confidence intervals described in \eqref{normalCI}.  
  In addition, we plot the estimates
  and confidence intervals described in \eqref{naiveNormalCI} which we term 
  \emph{naive}. These were not adjusted for selection.
  
  As we had seen in the univariate case, the conditional estimator
  acts as an adaptive shrinkage estimator. When the observed value is
  far away from the threshold, then no shrinkage is performed and when
  it is relatively close to the threshold then it is shrunk towards
  zero. 
\end{example}

\section{Maximum Likelihood Estimation for the Lasso}\label{sec:lassoMLE}
In this section we demonstrate how the ideas from the previous
section can be adapted for computing the post-selection MLE in linear
regression models selected by the Lasso. The Lasso estimator minimizes
the squared error loss augmented by an $\ell_1$ penalty,
\begin{equation*}
\hat\beta_{\text{Lasso}} = \arg\min_\beta\frac{1}{2}\|y - \bX\beta\|^{2}_2 + \lambda\|\beta\|_1
\end{equation*}
with $\lambda\geq 0$ being a tuning parameter.  Model selection results
from the fact that the $\ell_1$ penalty may shrink a subset of the
regression coefficients to zero.  As in \citet{Lee16}, we
perform inference on the non-zero regression coefficients in the Lasso
solution, that is, the selection procedure is $
S(y)\;  =\; \{j: \; \hat\beta_{\text{Lasso},j} \neq 0 \}
$.

Given selection of a model $M$, we are interested in estimating the
unconditional mean of the regression coefficients
$\beta = (\bX^{T}_M\bX_M)^{-1}\bX^{T}_M\E(Y)$.  We begin by describing the
Lasso selection event (Section \ref{sec:lassoSelect}) and then give a
Metropolis-Hastings sampler for the post-selection
distribution of the least-squares estimates
(Section \ref{sec:lassoMH}).  In Section \ref{sec:LassoSGD}, we
describe a practical stochastic ascent algorithm for estimation after
model selection with the Lasso.

\subsection{The Lasso Selection Event}\label{sec:lassoSelect}
Let $M\subseteq\{1,\dots,p\}$ be a given model.  In order to
develop a sampling algorithm for a normal distribution truncated to
the event that $S(\bX, y):= \{j: \, \hat\beta_{\text{Lasso},j} \neq 0 \}=M$, we
invoke the work of \citet{Lee16} who provide a useful
characterization of this Lasso selection event.  Let $s \in \{-1, 1\}^{|M|}$ be the vector
of signs of $\hat\beta_{\text{Lasso}}$ over the active set.  We will consider two sets
\begin{equation}
  \label{oneCondition}
  A_1(M,s) := \left\{\bA_1(M,s) y < u_1(M,s) \right\}, 
\end{equation}
\begin{equation} \label{zeroCondition}
  A_0(M,s) := \left\{l_0(M,s) < \bA_0(M)y < u_0(M,s)\right\},
\end{equation}
where in the first event 
\begin{equation}\label{refittedLSE}
\bA_1(M, s) = -\diag(s) (\bX^{T}_M \bX_M)^{-1} \bX_M^{T}, \quad
u_1(M, s) = -\lambda\diag(s)(\bX_M^{T} \bX_M)^{-1} s,
\end{equation}
and in the second event
\begin{gather}\label{fullLassoset}
  \bA_0(M)  = \frac{1}{\lambda} \bX^{T}_{-M} (I -
  \bX_M(\bX_M^{T}\bX_M)^{-1}\bX_M^{T}) , \\
  \nonumber
  l_0(M, s) = -\mathbf{1} -
  \bX_{-M}^{T}\bX_M(\bX^{T}_M\bX_M)^{-1} s, \;\;\;\;\;\; u_0(M, s) = \mathbf{1}
  - \bX_{-M}^{T}\bX_M(\bX^{T}_M\bX_M)^{-1} s.
\end{gather}
Here, $\bX_M$ is the submatrix of the design matrix $\bX$ made up of the
columns indexed by the selected model $M$ and the columns in the submatrix $\bX_{-M}$ correspond to variables which were not selected. It can be shown that 
\begin{equation}\label{MLassoDefinition}
\{S(\bX, Y) = M\;\text{and sign vector equal to} \ s\} \;=\; A_0(M,s) \cap A_1(M,s).
\end{equation}

Suppose that $Y \sim (\bX \beta, \sigma^{2} \bI)$, then conditional score function for a model selected by the Lasso is given by
\begin{align*}
\sigma^{2} \frac{\partial}{\partial \beta} \log \mathcal{L}(\beta | M) &=
\bX_M^{T}y  - \E(\bX^{T}_MY|M) = \bX^{T}_My - \frac{\sum_{s} P(M, s) \E\left(\bX^{T}_M Y |A_1(M,s)\right)}
{\sum_s P(M,s)},
\end{align*}
where for a given set of signs $P(M, s) = P\left(A_0(M,s)\right)\times P\left(A_1(M,s)\right)$.

As in the normal means problem, parameters related to the set of
variables excluded from the model play a role in the conditional
likelihood. In the normal means problem we advocated excluding those
from the optimization of the conditional likelihood.  For the Lasso, we
similarly must compute a conditional expectation which is a function
of $\bA_0(M) \E(Y)$.  We again advocate for avoiding conditional
likelihood-based estimation of this quantity.  In computational
experiments we observed that the value of $\bA_0(M) \E(Y)$ tends to be
very small and rather well approximated by a vector of zeros. For
more on this and some numerical examples see Appendix B. 

In the next subsection, we devise an algorithm for sampling from the
post-selection distribution of the regression coefficients selected by
the Lasso without conditioning on the sign vector $s$.  The sampler
will operate by updating the two quantities
$$
\eta := (\bX^{T}_M \bX_M)^{-1} \bX^{T}_M y, \qquad
\xi := \frac{1}{\lambda} \bX_{-M}\left(I - \bX_M(\bX^{T}_M \bX_M)^{-1} \bX^T_M\right) y.
$$

\subsection{Sampling from the Lasso Post-Selection Distribution}\label{sec:lassoMH}

With a view towards Gibbs sampling, we examine the region where a
single regression coefficient may lie given the signs of all other
coefficients.  Let $j\in M$ be an arbitrary index.  
Denote by $s^{+j}$ and $s^{-j}$ vectors of signs where the signs for all
coordinates but $j$ are held constant and the $j$th coordinates are
set to either $1$ or $-1$, respectively.  A necessary condition for the
selection of $M$ is that
$\eta_j \leq \lambda(\bX^{T}_M\bX_M)^{-1}_{j,.}s^{-j}$ or $\eta_j \geq \lambda(\bX^{T}_M\bX_M)^{-1}_{j,.} s^{+j}$.
Ideally, we would be able to implement a Gibbs sampler that allows for the change of signs as we have done in Section \ref{sub:truncnormsamp} by setting 

\begin{equation}\label{upperlowerdef}
l_j = \lambda(\bX^{T}_M\bX_M)^{-1}_{j,.}s^{-j}, \;\;\;\;\;
u_j = \lambda(\bX^{T}_M\bX_M)^{-1}_{j,.} s^{+j}.
\end{equation}
However, an important way in which the Lasso selection event differs from the one described in Section \ref{sec:selectedMeans} is that when a single coordinate of $s$ is changed, the thresholds for all other variables change. Thus, in order for a single coordinate of $\eta$ to change its sign, all other variables must be in positions that allow for that. 

In order to explore the entire sample space (and sign combinations) we
propose a delayed rejection Metropolis-Hastings algorithm
\citep{tierney1999some, mira2001metropolis}. The algorithm works by
attempting to take a Gibbs step for each selected variable in turn.  If the proposed Gibbs step for the $j$th variable satisfies the constraints induced by the selection event then the proposal is accepted. Otherwise, we keep the proposal for the $j$th variable and make a global proposal for all selected variables keeping their signs fixed. We use the notation: 
$$
\eta \sim N_p(\beta, \bSigma_1), 
\qquad \beta = (\bX_M^{T}\bX_M)^{-1}\bX_M^{T}\E(Y), \qquad 
\bSigma_1 = \sigma^{2} (\bX^{T}_M \bX_M)^{-1},
$$$$
\xi \sim N(0, \bSigma_0), 
\qquad \bSigma_0 = \sigma^{2}\bA_0(M) \bA_0(M)^{T}.
$$

At some arbitrary iteration $t$, our  sampler first makes the draw
\begin{equation}
\xi^{t} \sim f(\xi |M, \eta^{t-1}).\label{eq:xit}
\end{equation}
This sampling task is quite simple in the sense that $\xi |M, \eta$ has a multivariate normal distribution constrained to a convex set.
Next, we make a proposal for each selected variable. For the $j$th selected variable we sample:
$$
r_j \sim f\left(\eta_j \middle| \{\eta_j < l_j\}\cup \{u_j < \eta_j\},\eta_1^{t},\dots,\eta_{j-1}^{t},\eta_{j+1}^{t-1},\dots,\eta^{t-1}_{p}\right),
$$
where $l_j$ and $u_j$ are as defined in \eqref{upperlowerdef}. If the
sign of $r_j$ differs from the sign of $\eta_j^{t-1}$, then we must
verify that $\xi^{t}$ from \eqref{eq:xit} satisfies the constraints imposed by the new set
of signs. If the constraints described in
(\ref{zeroCondition})
 are not satisfied, then the proposal is rejected.  If the proposal
 yields a point that satisfies both \eqref{zeroCondition} and 
(\ref{oneCondition})
then no further adjustment is necessary and the
acceptance probability is $1$ because the proposal is full conditional
\citep{Chib1995understanding}. On the other hand, if the proposed
point is not in the set from \eqref{oneCondition},  then 
a sign change has been performed and 
we must update the values for other coordinates. 

Denote by $\text{TN}(a, b, \mu, \sigma^{2})$ a univariate normal distribution with mean $\mu$ and variance $\sigma^{2}$ constrained to the interval $(a,b)$. For all variables $k \neq j$ we sample a proposal from the following distribution:
\begin{equation}\label{truncNormProposal}
r_k \sim \text{TN}(a_k, b_k, \eta_k^{t}, \sigma^{2}_{k,-k}),
\end{equation}
where $a_k=u_k$ and $b_k=\infty$ if $s_k^t=1$, and $a_k=-\infty$ and
$b_k=l_k$ if $s_k^t=-1$.
Note that in \eqref{truncNormProposal} $l_k$ and $u_k$ must be recomputed according to the proposed sign change. 

The Metropolis-Hastings algorithm in its entirety is described in
Algorithm \ref{LassoMH} in the Appendix. The following Lemma describes the transitions
of the proposed sampler.

\begin{lemma}\label{delayedRejection}
For the $j$th variable at the $t$th iteration define:
\begin{flalign*}
&r_{1}^{\rightarrow} =\left(\eta_{1}^{t},\dots,\eta^{t}_{j-1},r_j,\eta^{t-1}_{j+1},\dots,\eta^{t-1}_{p}, \xi^t \right), \quad
r_{2}^{\rightarrow} = \left(r_1,\dots,r_{j-1},r_j,r_{j+1},\dots,r_p, \xi^t \right), 
\\
&r_1^{\leftarrow} =  \left(r_1,\dots,r_{j-1},\eta^{t-1}_j,r_{j+1},\dots,r_p, \xi^t \right), \quad
r_2^{\leftarrow} = \left(\eta_{1}^{t},\dots,\eta^{t}_{j-1},\eta_{j}^{t-1},\eta_{j+1}^{t-1},\dots,\eta^{t-1}_{p}, \xi^t \right).
\end{flalign*}
Here, $r_2^{\leftarrow}$ represents the current state of the sampler
after the Gibbs step from~(\ref{eq:xit}).
If $\xi^{t}$ from \eqref{eq:xit} is not in the set from 
(\ref{zeroCondition}), %% \eqref{fullLassoset} 
then 
the proposal for $r_j$ is rejected and the sampler stays in state
$r_2^{\leftarrow}$.  If $\xi^{t}$ is in (\ref{zeroCondition})
%% \eqref{fullLassoset}
 and $r_1^{\rightarrow}$ is in the set from \eqref{oneCondition} then the sampler moves to $r^{\rightarrow}_1$.  Otherwise, if
$r_1^{\leftarrow}$ is in the set from \eqref{oneCondition} then the
sampler stays in state $r_2^{\leftarrow}$. Finally if neither
$r^{\rightarrow}_1$ nor $r^{\leftarrow}_1$ are in the set from
\eqref{oneCondition} then the 
sampler either moves to $r_{2}^{\rightarrow}$ or stays put at
$r_2^{\leftarrow}$.  In this case, the move to $r_{2}^{\rightarrow}$
occurs with probability
\begin{equation}\label{MHtransition}
p^t_j = \min \left(\frac{\varphi(r_2^{\rightarrow} ; \beta, \bSigma)}{\varphi(r^{\leftarrow}_2; \beta, \bSigma)} 
\frac{q(r^{\rightarrow}_2,r^{\leftarrow}_2)}{q(r^{\leftarrow}_2,r^{\rightarrow}_2)}
, 1\right),
\end{equation}
where
$$
q(x, y) =  f\left(y_j \middle| \{Y_j < l_j\} \cup \{u_j < Y_j\},x_{-j}\right)
\prod_{k\neq j} \frac{\varphi(y_k ; x_k, \sigma^{2}_{k,-k})}{P(Y_k \in (a_k, b_k); x_k, \sigma^{2}_{k,-k})}.
$$  
\end{lemma}

\subsection{A Stochastic Ascent Algorithm for the Lasso}\label{sec:LassoSGD}

We now propose an algorithm for computing the post-selection MLE when
the model is selected via Lasso. We begin by defining the gradient
ascent step, which uses samples from the post-selection distribution
of the refitted regression coefficients.  We give a convergence
statement for the resulting algorithm, and we discuss practical
implementation for which we address variance estimation and imposing sign constraints.

Let $M=S(y)$ be the Lasso-selected model.  Given a sample
$\eta^i \sim f_{\hat\beta^{i-1}}(\eta|M)$ from the post-selection
distribution of the least squares estimator, we take steps of the
form:
\begin{equation}\label{Lassosteps}
\hat\beta^{i} = \hat\beta^{i-1} + \gamma_i\left( \bX_M^{T}y - (\bX^{T}_M \bX_M)\eta^{i} \right),
\end{equation}
where the $\gamma_i$ satisfy the conditions from
\eqref{gammaconditions}.  In Theorem \ref{LassoSGDconvergence} we give
a convergence statement for the algorithm defined by
\eqref{Lassosteps}. As in Theorem \ref{thm:nmeans}, the main challenge
is bounding the variance of the stochastic gradient steps.

\begin{theorem}\label{LassoSGDconvergence}
  Let $\eta$ follow the conditional distribution of
  $\eta \sim N(\beta, \bSigma)$ given the Lasso selection $S(y)=M$.
% , and let $M=S(y)$ be the
%   Lasso-selected model. 
Then there exists a constant $A$ such that for
  all $\beta \in \mathbb{R}^{p}$:
$$
\E_\beta\left\|(\bX^{T}_M\bX_M)\eta - \bX_M^{T}\E_\beta(Y|M)\right\|^{2}_2 \leq A.
$$
Furthermore, the sequence $(\hat\beta^i)$ from~(\ref{Lassosteps})
converges, and its limit
$\hat\beta^{\infty} := \lim_{i\rightarrow\infty} \hat\beta^{i}$ satisfies $\psi(\hat\beta^\infty) := \bX_M^{T}y - \E_{\hat\beta^\infty}(\bX^{T}_M Y|M) = 0$.
\end{theorem}

Before we exemplify the behavior of the proposed algorithm we first
discuss some technicalities. The sampling algorithm proposed in the
previous section assumes knowledge of the residual standard error $\sigma$,
a quantity that in practice must be estimated from the data. We find that
the cross-validated Lasso variance estimate recommended by \citet{reid2013study} works well for our purposes.

As in the univariate normal case, the post-selection estimator for the
Lasso performs adaptive shrinkage on the refitted regression
coefficients.  However, the asymmetry between the thresholds dictated by different sign sets   
may cause the sign of the conditional coefficient estimate to be different than the one inferred by the Lasso. Empirically we have found some benefit for
constraining the signs of the estimated coefficients to those of the
refitted least-squares coefficient estimates.

\begin{figure}[t]\label{lassoExampleFig}
\includegraphics[width=5.5 in]{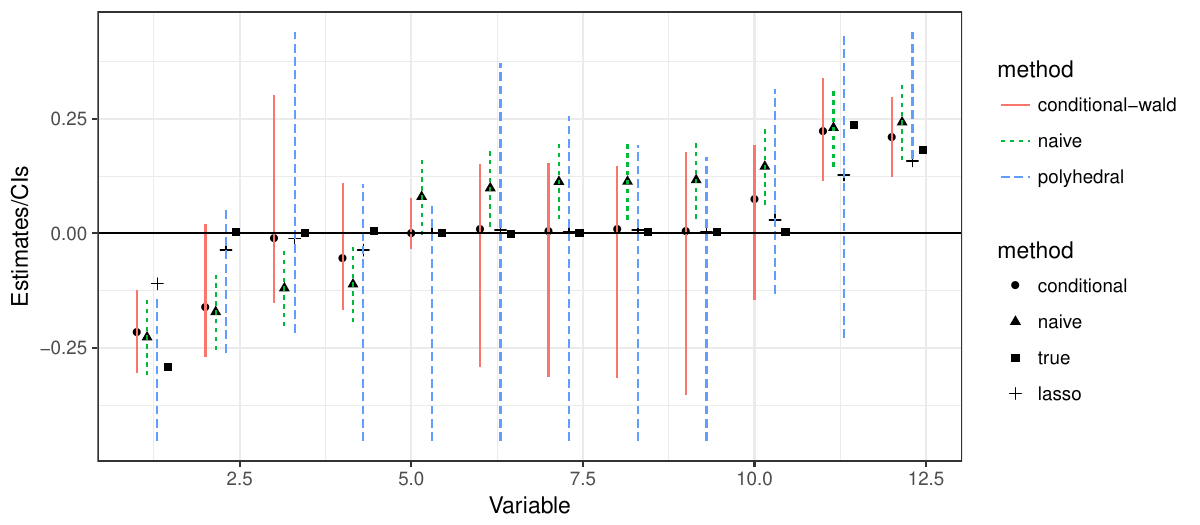}
\caption[]{Post-selection estimates and confidence intervals for the
  Lasso.  For simulated data, we plot the 
  conditional MLE (circles), refitted least-squares estimates
  (triangles) and Lasso estimates (squares). The true coefficient
  values are marked by plus signs.  We also plot three types of
  confidence intervals, Conditional-Wald (solid red line),
  Refitted-Wald (dashed green line) and Polyhedral confidence
  intervals (dashed blue lines).}
\label{lassoExampleFig}
\end{figure}

\begin{example}\label{LassoExample}
  We illustrate the proposed method via simulated data
  that are generated as follows.  We form a matrix of covariates by
  sampling $n$ rows independently from $N_p(0,\bSigma)$ with
  $\bSigma_{i,j} = \rho^{|i-j|}$.  We then generate a coefficient
  vector $\beta$ by sampling $k$ coordinates from the
  $\mathit{Laplace}(1)$ 
distribution and setting the rest to zero.
  Next, we sample a response vector $Y \sim N(\mu, \sigma^{2} \bI)$,
  where $\mu = \bX\beta$ and $\sigma^{2}$ is chosen to obtain a certain
  signal-to-noise ratio defined as $\snr := \Var(\mu)/\sigma^{2}$.  We
  set $n = 400$, $p = 1000$, $k = 5$, $\rho = 0.3$ and $\snr = 0.2$.

  Given a simulated dataset we select a model using the Lasso as
  implemented in the R package `glmnet'
  \citep{friedman2010regularization}.  Following common practice and
  the default of the package, the tuning parameter $\lambda$ is
  selected via cross-validation.  Strictly speaking, this
  yields another post-selection problem.

  In Figure \ref{lassoExampleFig} we plot three types of estimates for
  the regression coefficients selected by the Lasso. The conditional
  estimator proposed here, the refitted least-squares estimates and the
  Lasso estimates. In addition to the point estimates, we also plot
  three types of confidence intervals. The first are the
  Conditional-Wald confidence intervals analogous to the ones
  described in Section \ref{sub:normCI}.  They are given by:
  \begin{gather*}
 \hat{\CI}_j = \left(\hat\beta^{M}_{j,n} - \hat\TN_{j, 1-\alpha/2} /
   \sqrt{n}, \; \hat\beta^{M}_{j,n} - \hat\TN_{j, \alpha/2} /
   \sqrt{n}\right),
 \\
 \hat{\TN} =^D  \sigma^{2}
 \Var_{\hat\beta_n^M}\left(n^{-0.5} \bX^{T}_{M}Y\middle|M\right)^{-1}
  n^{-0.5}\left(\bX^{T}_My - \E_{\hat\beta^{M}_n}(\bX^{T}_M Y|M)\right).
\end{gather*}
The second intervals are the
Refitted-Wald confidence intervals obtained from fitting a linear
regression model to the selected covariates without accounting for
selection. Finally, we also include the intervals of
\citet{Lee16} as implemented in the R package
`selectiveInference' \citep{tibshirani2016selectiveInference}.  We
term these \emph{Polyhedral} confidence intervals.

  In Figure \ref{lassoExampleFig}, black circles mark the conditional
  estimates, triangles the refitted least squares estimates, squares the
  lasso estimates and plus signs the true coefficient values. The
  conditional estimator tends to lie between the refitted and the lasso
  estimates.  When the refitted estimate is far from zero the
  conditional estimator applies very little shrinkage, and when the refitted
  estimator is closer to zero the conditional estimator is shrunk towards
  the lasso estimate.  The conditional confidence intervals also
  exhibit a behavior that depends on the estimated magnitude of the
  regression coefficients.  When the conditional estimator is far from
  zero the size of the confidence intervals is similar to the size of
  the refitted confidence interval. When the conditional estimator is
  shrunk towards zero, its variance tends to be the smallest. The
  confidence intervals are the widest when the conditional estimator
  is just in-between the lasso and refitted estimates. The Polyhedral
  confidence intervals tend to be the largest in most cases.  Section
  \ref{simStudy} gives a more thorough
  examination of these estimates and confidence intervals.
\end{example}

\section{Asymptotics for Conditional Estimators}
\label{sec:asympt-cond-estim}

We now present asymptotic distribution theory that supports the
estimation method proposed in the previous sections.  Such theory is complicated
by the fact that model selection induces dependence between the
previously i.i.d.\ observations.  In Section \ref{sub:consistency} we
first give a consistency result for naive unconditional estimates,
which in particular justifies our plug-in likelihood method for the
normal means problem.  We then outline conditions under which the
conditional MLE is consistent for the parameters of interest in a general exponential family
setting.  In Section
\ref{sub:asymplasso} we adapt the theory to the Lasso post-selection
estimator.  
We remark that theory on the efficiency of
conditional estimators can be found in
\citet{Routenberg15}.  Proofs for this section are
deferred to the appendix.

\subsection{Theory for exponential families}\label{sub:consistency}

Suppose we have an i.i.d.~sequence of observations
$(Y_i)_{i=1}^\infty$ drawn from a distribution $f^*$.  As a base model
for the distribution of each observation $y_i$, consider a regular
exponential family $\{p_\theta:\theta\in\Theta\}$ with sufficient
statistic $T \in \mathbb{R}^{p}$ and natural parameter $\theta$.  So, $\Theta\subset\mathbb{R}^p$. For the sample
$y_1,\dots,y_n$, define $\bar{T}_n := n^{-1} \sum_{i=1}^{n} T(y_i)$.
Now, let $\mathcal{M}$ be a countable set of submodels, which we
denote by $M = \{p_{\theta^M} : \theta^M \in \Theta^M\}$ with
parameter space $\Theta_M\subset\Theta$.  We consider a model
selection procedure $S_n : \mathbb{R}^{p} \rightarrow \mathcal{M}$ that
selects a model $M$ as a function of $\bar T_n$.  Based on the true
distribution $f^*$ the sample is taken from, the selection procedure
$S_n$ induces a distribution $P_n(M):=P(S_n(\bar T_n) = M)$ over
$\mathcal{M}$.  We emphasize that $f^*$ need not belong to any model
in $\mathcal{M}$ nor the base family $\{p_\theta:\theta\in\Theta\}$.

\begin{example}
  In the normal means problem, $p_\theta$ is a normal distribution
  with mean vector $\theta$.  The sufficient statistic is
  $T(y)=\bSigma^{-1}y$, where $\bSigma$ is the known covariance matrix.
  Each model $M\in\mathcal{M}$ corresponds to a set of mean vectors
  with a subset of coordinates equal to zero.  The selection procedure
  $S_n$ is based on comparing the coordinates of $\bar T_n$ to
  predetermined thresholds $l_{j}$ and $u_j$,
  recall~(\ref{mvnSelection}). In an asymptotic setting $l_j$ and $u_j$ will 
  often scale with the sample size to obtain a pre-specified type-I error rate. 
\end{example}

We consider estimation of a parameter $\theta^M_0$ of a fixed model
$M$, which represents the model selected in the data analysis.  If the
data-generating distribution $f^*$ belongs to $M$, then
$f^*=p_{\theta^M_0}$ for a parameter value $\theta^M_0 \in \Theta^M$
and consistency can be understood as referring to the true
data-generating distribution.  If $f^*\not\in M$, then the parameter
in question corresponds to the distribution in $M$ that minimizes the
KL-divergence from $f^*$, so
$$
\theta^M_0 := \arg\inf_{\theta^M\in\Theta^M} -\E_{f^*} \left[\log p_{\theta^M}(Y) - \log f^*(Y) \right] =
\arg\sup_{\theta^M\in\Theta^M} \E_{f^*} \left[\ell_{\theta^M}(Y)\right].
$$
Note that even under model misspecification we have $
\E_{f^*}(\bar{T}_n) = \E_{\theta_0^M}(\bar{T}_n)
$
because $\theta^M_0$ is the solution to the expectation of the score
equation.

The post-selection setting is unusual in the sense that we are only
interested in a specific model $M$ if $S_n(\bar T_n) = M$.  Hence, it
only makes sense to analyze the asymptotic properties of an estimator
of $\theta^M_0$ if model $M$ is selected infinitely often as
$n\rightarrow \infty$.  This justifies our subsequent focus on
conditions that involve the probability of selecting $M$.

Our first result applies in particular to the normal means problem and
is concerned with the post-selection consistency of the
unconditional/naive MLE for $\theta^M_0$.

\begin{theorem}\label{thm:naiveconsistency}
  Let $M$ be a fixed model with $P_n(M)^{-1}e^{-\delta n} = o(1)$ for
  all $ \delta>0$.  Let
  $\tilde\theta_n^M=(\tilde\theta^M_{n,j})_{j=1}^p$ be an estimator
  that unconditionally is unbiased for $\theta^M_0$.  Suppose
  there is a constant $C\in(0,\infty)$ such that for all $1\le j\le p$
  and $n\ge 1$ the distribution of
  $\sqrt{n}(\tilde\theta^M_{n,j}-\theta^M_{0,j})$ is sub-Gaussian for
  parameter $C$.  Then $\tilde\theta^M_n$ is post-selection
  consistent, that is,
  $$
  \lim_{n\rightarrow \infty} P(\|\tilde\theta^M_n -
  \theta^M_0\|_{\infty} > \varepsilon \,|\,  S_n(\bar T_n) = M ) =0 \qquad \forall \varepsilon > 0.
  $$
\end{theorem}

Next, we turn to the conditional MLE.  Let $\ell_{\theta^M}(y_i)$ be
the log-likelihood of $y_i$ as a function of $\theta^M$, and let
$P_{n,\theta^M}(M)$ be the probability of $\{S_n(\bar T_n) = M\}$ where
$y_1,\dots,y_n$ is an i.i.d.~sample from $p_{\theta^M}$.  Then the
conditional MLE is
$$
\hat{\theta}^M_n = \arg\max_{\theta^M \in \Theta^M} \left( \frac{1}{n} \sum_{i=1}^{n}  \ell_{\theta^M}(y_i) \right) - \frac{1}{n}\log P_{n,\theta^M}(M).
$$
We now give conditions for its post-selection consistency.

\begin{theorem}\label{thm:consistency}
  Suppose the fixed model $M$ satisfies 
  \begin{gather}\label{true-o(n)}
    P_{n}(M)^{-1} = o(n),\\
  \label{exp-inf-prob}
\lim_{n\rightarrow\infty}\inf_{\theta^M} P_{n,\theta^M}(M)e^{n} =\infty.
\end{gather}
Furthermore, suppose that for a sufficiently small ball $U\subset \Theta $ centered at $\theta^M_0$
\begin{equation}\label{newSup}
\sup_{\theta^M\in U(\theta^M_0)} P_{n,\theta^M}(M)^{-1} = o(n).
\end{equation}
Then the conditional MLE is post-selection consistent for $\theta^M_0$, that is,
$$
\lim_{n\rightarrow\infty} P(\|\hat\theta^M_n - \theta^M_0\|_\infty > \varepsilon \,|\,  S_n(\bar{T}_n) = M ) = 0\qquad \forall \varepsilon > 0. 
$$
\end{theorem}

Condition \eqref{exp-inf-prob} concerns the model-based selection
probability and ensures that the conditional MLE exists with
probability $1$ as $n\rightarrow \infty$. Both the plug-in likelihood
for the selected means problem and the Lasso likelihood satisfy this
condition.  We note that this condition excludes examples such as the
singly truncated univariate normal distribution, where the probability
that an MLE does not exist is positive \citep{castillo1994singly}.
Condition \eqref{true-o(n)} concerns the true probability of selecting
the considered model $M$, which is required to not
decrease too fast. Condition \eqref{newSup} serves to ensure that the conditional score function is well behaved in the neighborhood of the estimand. 

\subsection{Theory for the Lasso}\label{sub:asymplasso}
In this section we describe how the theory from the previous
section applies to inference in linear regression after model
selection with the Lasso. Suppose that we observe an independent
sequence of observations
\begin{equation}\label{normalityAssumption}
(Y_i)_{i=1}^{\infty} \sim N(\mu_i,\sigma^{2}). 
\end{equation}
Each observation $Y_i$ is accompanied by a vector of covariates 
$X_i\in \mathbb{R}^{p}$ which we consider fixed, or equivalently, 
conditioned upon. The sufficient statistic
for the linear regression model is given by $T_n(\bX, y)  = \bX^{T}y$ 
and the model selection function $S_n(\bX,y)$ is the Lasso, which selects a model:
$$
S_n(\bX,y) = \{j: \; \hat\beta_{\text{Lasso},j} \neq 0 \}.
$$

For a selected model $M$, the conditional MLE for the regression coefficients is given by:
\begin{equation}\label{LassoMLEdef}
\hat\beta^M_n = \arg\max_{\beta} \frac{f(\bA_1 y)}{P_\beta(M)},
\end{equation}
where $
P_\beta(M) = \sum_{s} P_\beta(\bA_1(M,s)) \times P_n(\bA_0(M, s))$.
Notice that in our objective function the probabilities for not selecting the null-set are not a function of the parameters over which the likelihood is maximized. Instead, they are defined as a function of the sample size $n$ and are determined by the imputed value for $\bA_0(M)\mu$. In practice we set $\bA_0(M)\mu = 0$. This imputation method can be justified by the fact that a model is unlikely to be selected infinitely often if
$
\lim_{n\rightarrow\infty} \bA_0(M)\mu \neq 0
$.

For good behavior of the conditional MLE we made assumptions regarding
the probabilities of selecting models of interest. Many previous works
have investigated the properties that a data generating distribution
must fulfill in order for the Lasso to identify a correct model with
high probability. See for example \citet{zhao2006model}, and
\citet{meinhausen2009Lasso}. While we do not limit our attention to
the selection of the correct model, this line of study sheds light on
the conditions that any model $M\in\mathcal{M}$ must satisfy in order 
to be selected with sufficiently high probability. In the following we assume that the number of covariates $p_n = p$ is kept fixed while the sample size $n$ grows to infinity. We touch on high-dimensional settings briefly at the end of the section. 

The set of models for which we are able to guarantee convergence depends on the scaling of the $\ell_1$ penalization parameter. We consider two types of scalings:
\begin{equation}\label{slowscaling}
\lambda_n \propto \sqrt{n},
\end{equation}
\begin{equation}\label{fastscaling}
\lim_{n\rightarrow\infty}\frac{\lambda_n}{\sqrt{n}} = \infty, \qquad
\lim_{n\rightarrow\infty}\frac{\lambda_n}{n} = 0.
\end{equation}

We begin by discussing the case where the $\ell_1$ penalization parameter scales as in \eqref{slowscaling}. In this setting, the model selection probabilities can be bounded in a satisfactory manner as long as the expected projection of the model residuals on the linear subspace spanned by the inactive variables is not too large. 
\begin{lemma}\label{slowLemma}
Suppose that $\lambda_n$ scales as in \eqref{slowscaling} and that $y$ follows a normal distribution as defined in \eqref{normalityAssumption}.  Suppose further that for an arbitrary model of interest $M\in\mathcal{M}$ there is a matrix $\bSigma$ and a vector $\beta_0^M$ such that following holds:
\begin{equation}\label{slowConditions}
\frac{1}{n} \bX^T \bX \rightarrow\bSigma,
\end{equation}
\begin{equation}\label{meanConditions}
\qquad (\bX^T_M \bX_M)^{-1} \bX_M^{T}\mu\rightarrow\beta_0^M, 
\qquad \bA_0(M)\mu \rightarrow 0, \;\; a.s.
\end{equation}
Then there exists an asymptotic lower bound for the probability of selecting $M$:
$$
\lim_{n\rightarrow\infty} P_n(M) \geq \lim_{n\rightarrow\infty}\inf_{\beta^M} P_{n,\beta^M}(M) = c > 0.
$$
\end{lemma}  

Next, we discuss the  
setting where $\lambda_n$ grows faster than $\sqrt{n}$. Here we must
impose stronger conditions on the selected model because  the
probability of selecting a model which contains covariates with zero
coefficient values may decrease to zero at an exponential rate. Furthermore, we make assumptions similar to the Irrepresentable Conditions of \citet{zhao2006model} on the selected model in order to make sure that the model selection conditions corresponding to the variables not included in the model are satisfied with high probability. We emphasize that we do not assume that the Irrepresentability Conditions hold in order to satisfy the selection of a true model, rather, we make these assumptions in order to identify models (correct or not) for which we can guarantee the consistency of our estimators. 
\begin{lemma}\label{fastLemma}
Suppose that $\lambda_n$ scales as in \eqref{fastscaling} and that conditions \eqref{normalityAssumption} and \eqref{meanConditions} hold. Furthermore, assume that:
$$
\frac{1}{n} \bX^{T}_M\bX_M\rightarrow\bSigma_M, \;\; a.s., \qquad
|\beta^M_{0j}| > 0, \;\forall j\in M, 
$$
and that
\begin{equation}\label{irrep}
\lim_{n\rightarrow\infty}\sup |\bX^{T}_{-M}\bX_M(\bX^{T}_MX_M)^{-1}s| \leq \nu < \boldsymbol{1}, \;\; \forall  s\in\{0,1\}^{|M|},
\end{equation}
for some constant $\nu$, where $\boldsymbol{1}$ is a vector of ones and the inequality holds element wise. Under these conditions the following limits hold:
$$
\lim_{n\rightarrow\infty}\inf_{\beta^M} P_{n,\beta^M}(M)e^n = \infty,
\qquad \lim_{n\rightarrow\infty} \inf_{\beta^M \in U(\beta^M_0)} P_{n,\beta^M}(M) = 1.
$$
\end{lemma}

The linear regression model trivially satisfies the modeling assumptions we made in the previous section. Thus, under the conditions given in the lemmas stated in this section, the conditional MLE for a model selected by the Lasso can be guaranteed to be well behaved.
\begin{corollary}\label{Lassothm}
Fix a model $M\in \mathcal{M}$ and suppose that the conditions of either Lemma \ref{slowLemma} or Lemma \ref{fastLemma} are satisfied. Then the conditional MLE \eqref{LassoMLEdef} is consistent for $\beta^M_0$.
\end{corollary}

\begin{remark}[High-Dimensional Problems]
The Lasso is often used in cases where the number of covariates $p$ is
much larger than $n$. In order to make asymptotic analysis relevant to
such cases it is common to assume that $p$ grows with the sample
size. While the theory developed here does not explicitly treat such a high-dimensional setting, none of our assumptions prevent us from allowing the model selection function $S_n$ to consider a growing number of covariates as $n$ grows. Specifically, if we assume that the $\ell_1$ penalty scales at the rate of 
$
\lambda_n = O\left( \sqrt{n\log p_n} \right)
$
as prescribed e.g.~by \citet{hastie2015statistical}, then our theory
applies as long as the assumptions of Lemma \ref{fastLemma} are satisfied and $\log p _n = o(n)$. 
\end{remark}

\begin{remark}[Normality]
While we made a simplifying normality assumption, we expect that for fixed dimension $p$, non-normal errors can be addressed using conditions similar to those outlined by
\citet{tibshirani2015uniform}. For theory for selective inference with
non-normal errors in the high-dimensional case, see the work of
\citet{Tian16}.
\end{remark}

\section{Simulation Study} \label{simStudy}

In order to more thoroughly assess the performance of the proposed
post-selection estimator for the Lasso, we perform a simulation study,
which we pattern after that in \citet{meinhausen2007relaxed}. 
We consider prediction and coefficient estimation using Lasso, our
conditional estimator and refitted Lasso.  We note already
that 
while some existing theoretical works outline conditions under which the refitted Lasso should outperform the Lasso in prediction and estimation \citep{lederer2013trust}, this does not occur in any of our simulation settings. For confidence intervals we compare our Wald confidence intervals to the confidence intervals of \citet{Lee16} which we term \emph{Polyhedral}. We find that both selection adjusted methods achieve close to nominal coverage rates. 

We generate artificial data for our simulations in a similar manner as we have done for Example \ref{LassoExample} in Section \ref{sec:LassoSGD}. We vary the sample size $n= 100, 200, 400, 800$, signal-to-noise ratio $\snr = 0.2, 0.8$, and the sparsity level $k = 2, 5, 10$. For each combination of parameter values we generate data and fit models $400$ times.  We keep the amount of dependence fixed at $\rho = 0.5$ and the number of candidate covariates fixed at $p = 400$. 

\begin{figure}[t]
\begin{center}
\includegraphics[width=5 in]{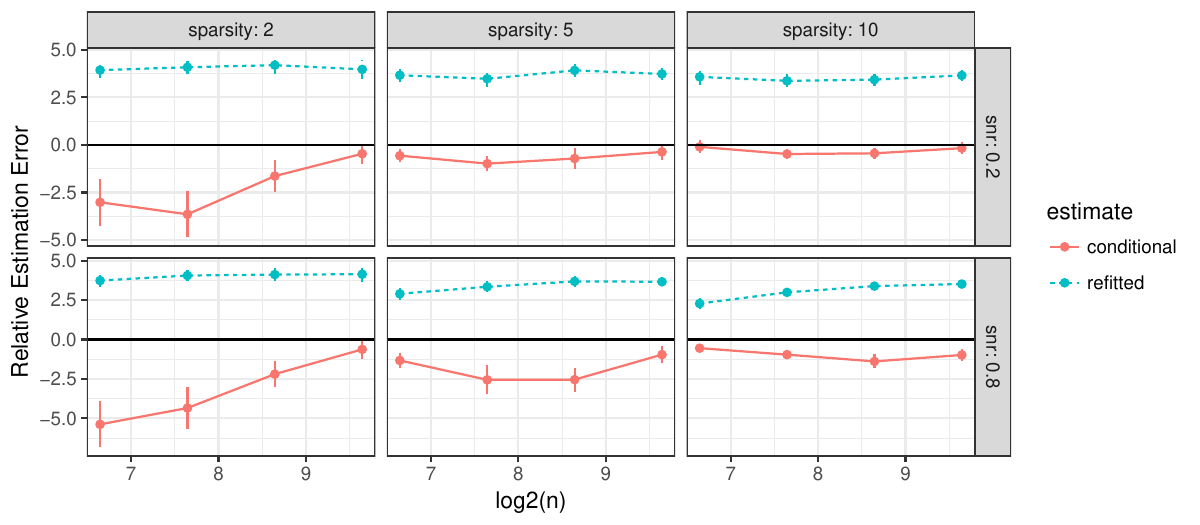} 
\caption[Estimation error of the conditional MLE for Lasso regression coefficients]{The relative estimation error of the regression coefficients compared to the Lasso as defined in \eqref{relativeMSEdef}. The error of the conditional estimates  (solid red line) is lower than that of the Lasso in all simulation settings and the error of the refitted least-squares estimates (dashed blue line) was worse than that of the Lasso in all simulations.}
\label{LassoMSEfig}
\end{center}
\end{figure}

In Figure \ref{LassoMSEfig} we plot the log relative estimation error of the refitted-Lasso estimates and the conditional estimates compared to the Lasso as defined by:
\begin{equation}\label{relativeMSEdef}
\frac{1}{|M|}\left(\sum_{j\in M}\log_2(\hat\beta_j - \beta_j) - \log_2(\hat\beta_{Lasso_j} - \beta_j)\right).
\end{equation}
This measure of error gives equal weights to all regression coefficients regardless of their absolute magnitude. In all simulation settings the refitted least-squares estimates are significantly less accurate than the Lasso or the conditional estimates. The conditional estimates tend to be more accurate than the Lasso estimates in all simulation settings. The conditional estimate tends to do better when there are at least some large regression coefficients in the true model. 

\begin{figure}[t]
\begin{center}
\includegraphics[width=5 in]{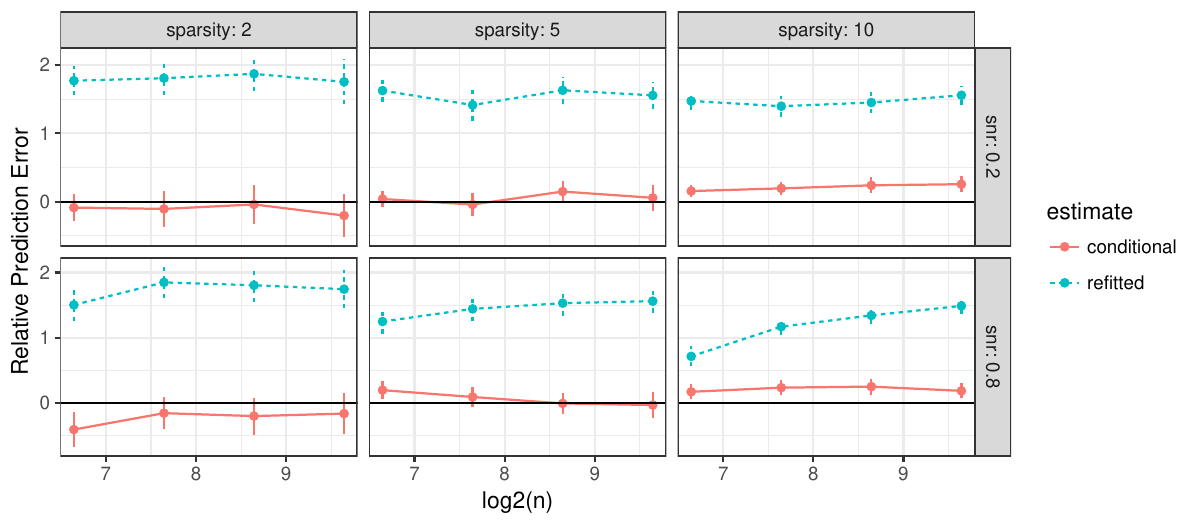} 
\caption[Prediction error of the conditional MLE compared to the Lasso]{The log of the ratio between the prediction errors for the conditional (solid red line) and refitted least-squares regression estimates (dashed blue line) relative to the prediction error of the Lasso as defined in \eqref{relativePredDef}. The conditional MLE produces better prediction than the Lasso when the signal is spread over a smaller number of variables.}
\label{LassoPrediction}
\end{center}
\end{figure}

In Figure \ref{LassoPrediction} we present the relative prediction error of the refitted least-squares Lasso estimates and the conditional estimates, as defined by:
\begin{equation}\label{relativePredDef}
\log_2\|\bX\hat\beta - \mu\|^{2}_2 - \log_2\|\bX\hat\beta_{\text{Lasso}} - \mu\|^{2}_2.
\end{equation}
Here, the Lasso provides more accurate predictions when the true model has more non-zero coefficients and the conditional estimator tends to be more accurate when the true model is sparse. 

\begin{figure}[t]
\begin{center}
\includegraphics[width= 5 in]{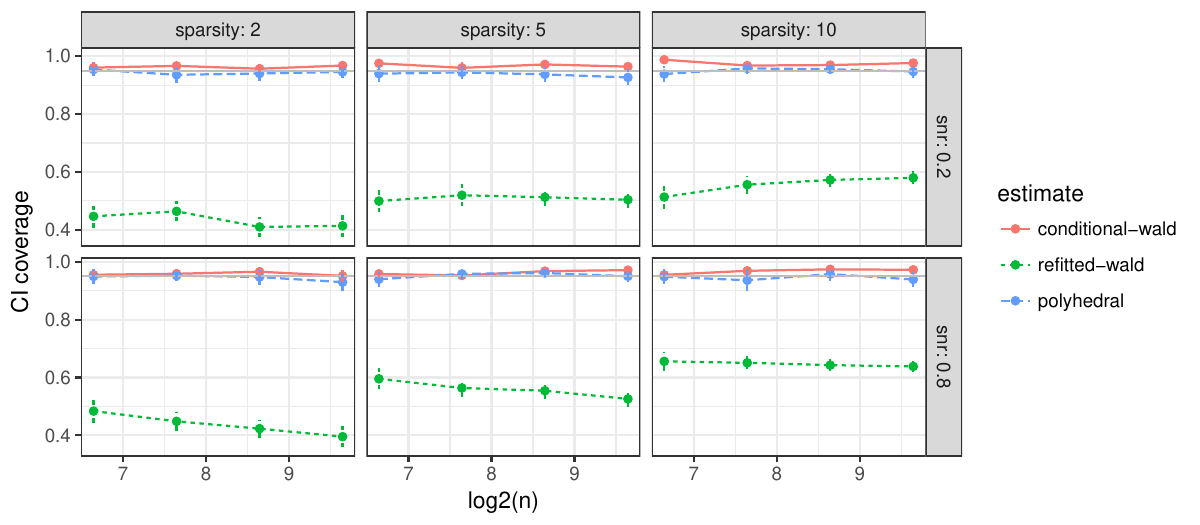} 
\caption[Coverage rate of post-selection confidence intervals for the Lasso]{Confidence interval coverage rate after model selection. Both the Conditional Wald CIs (solid red line) and the Polyhedral CIs (dashed blue line) achieve the target coverage rate of $95\%$ (horizontal grey line). The coverage rate of the unadjusted Wald confidence intervals (dotted green line) is far below nominal.}
\label{LassoCover}
\end{center}
\end{figure}

In Figure \ref{LassoCover} we plot the coverage rates obtained by the Conditional-Wald confidence intervals proposed here, the Polyhedral confidence intervals and the refitted `naive' confidence intervals. Both  of the selective methods obtain close to nominal coverage rates. 
The coverage rates of the refitted confidence intervals which were not adjusted for selection were far below the nominal levels in all simulation settings. 

While the two types of selection adjusted confidence intervals seem to
be roughly on par with respect to their coverage rate, they tend to
differ in their size. For Figure \ref{LassoCIsize} we 
generate the additional datasets with a smaller number of candidate covariates $p = 200$, a larger range of sample sizes- $n = 40, 75, 150, 300, 600, 1250, 2500, 5000, 10000$, a signal-to-noise ratio of $\snr = 0.2$ and $k = 10$ non-zero regression coefficients. 

We face some difficulty in assessing the average size of the Polyhedral confidence intervals, as these sometimes have an infinite length. a measure for the length of a typical confidence interval, we take the median confidence interval length in each simulation instance. In Figure \ref{LassoCIsize} we plot boxplots describing the distribution of the log relative size of the selection adjusted confidence intervals to that of the unadjusted refitted confidence intervals which tend to be the shortest. We find that as the sample size increases, the sizes of the Conditional-Wald confidence intervals are roughly twice the size the unadjusted confidence intervals, while the typical size of a Polyhedral interval is about twice the size of the Conditional-Wald confidence interval.

\begin{figure}[t]
\begin{center}
\includegraphics[width= 4.5 in]{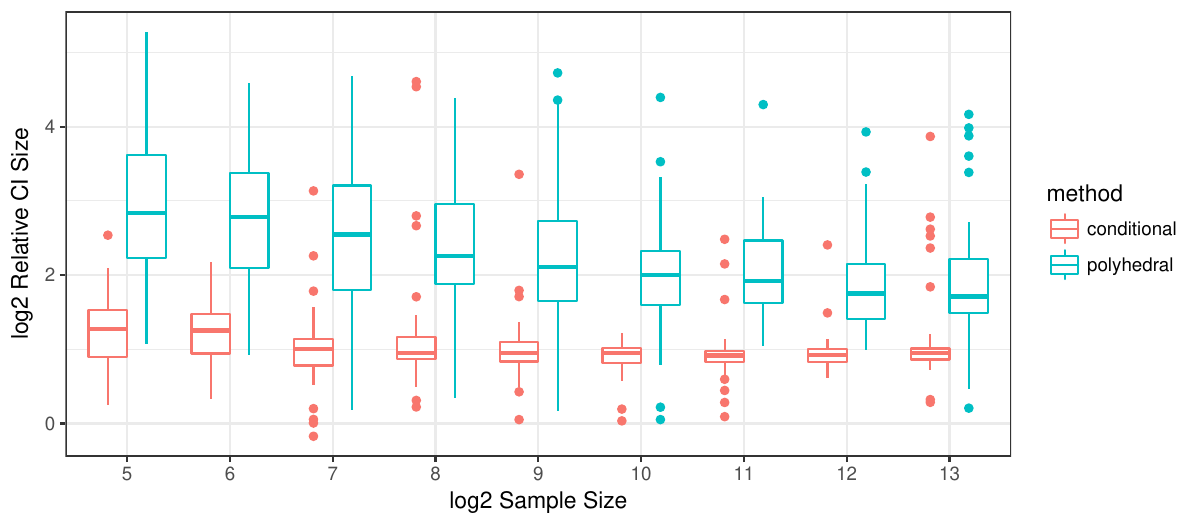} 
\caption[Comparison of confidence interval sizes after model selection with the Lasso]{Boxplots of the relative median sizes of the selection adjusted confidence intervals to the in relations to the unadjusted ones. The Conditional-Wald confidence are much shorter than the Polyhedral ones under all simulation settings and their size are far less variable.}
\label{LassoCIsize}
\end{center}
\end{figure}

\section{Conclusion}
In this work we presented a computational framework which enables, for
the first time, the computation of correct maximum likelihood
estimates after model selection with a possibly large number of
covariates. We applied the proposed framework to the computation of
maximum likelihood estimates of selected multivariate normal means and
regression models selected via the lasso. 

Our methods take the arguably most ubiquitous approach to data analysis, that of computing maximum likelihood estimates and constructing Wald-like confidence intervals. Furthermore, we do not involve conditioning on information additional to the identity of the selected model. A practice which, as shown by \citet{Fithian15}, may lead to a loss in efficiency.

We experimented with the proposed estimators and confidence intervals
in a comprehensive simulation study. The proposed conditional
confidence intervals were shown to achieve conservative coverage rates
and the point estimates were shown to be preferable to the
refitted-least squares coefficients estimates in all simulation
settings, and preferable to the Lasso coefficient estimates when there are large signals in the data. 

While in this work we focused on inference in the linear regression method, our framework and theory are directly applicable to any exponential family distribution. Specifically, it is immediately applicable to estimation of parameters of selected generalized linear models using the normal approximations proposed by \citet{taylor2016generalized}.

\subsection*{Supplementary Material} 
Proofs of theorems can be found in Appendix A. Some numerical examples for different plug-in methods for the Lasso MLE are in Appendix B. Analysis for maximum likelihood inference after one-sided testing is in Appendix C. Pseudo-code  for the algorithms used in the paper is in Appendix D. A software package and example scripts can be found at: \emph{https://github.com/ammeir2/selectiveMLE}.

%\bibliography{mselectionBIB}
\bibliography{aggregateBib}
\bibliographystyle{apalike}

\appendix 
\section{Proof of theorems}
\subsection{Proof of Theorem \ref{thm:nmeans}}
In their work on the convergence of stochastic gradient methods, \citet{Bertsekas00} formulate a general stochastic gradient method as an iterative optimization method consisting of steps of the form:
$$
x_{t+1} = x_{t} + \gamma_t(s_t + w_t),
$$
where $\gamma_t$ satisfies the condition from \eqref{gammaconditions},
$s_t$ is a deterministic quantity related to the true gradient and
$w_t$ is a noise component. They outline conditions regarding $s_t$
and $w_t$ that ensure the convergence of the ascent algorithm to an
optimum of a function $f(x)$ which possesses a gradient $\nabla
f(x)$. The conditions require that there exist positive scalars $c_1$
and $c_2$ such that for all $t$:
\begin{equation}\label{stConditions}
c_1\|\nabla f(x_t)\|^{2} \leq \nabla f(x_t)^{T} s_t, \;\;\;\;\; \|s_t\| \leq c_2(1 + \|\nabla f(x_t)\|),
\end{equation}
and that
\begin{align}\label{wtExp}
\E\left[ w_t  \,\middle|\, \mathcal{F}_t \right] &= 0,\\
\label{wtVar}
\E\left[\| w_t \|^{2} \,\middle|\, \mathcal{F}_t \right] &\leq A\left(1 + \|\nabla f(x_t)\|\right),
\end{align}
where $\mathcal{F}_t$ is the filtration at time $t$, representing all historical information available at time $t$ regarding the sequence $(w_{t},s_{t})_{i=1}^{\infty}$.

In our case, the function of interest is the conditional
log-likelihood $f(x) = l(\mu) := \log \mathcal{L}(\mu)$, where the
coordinates of $\mu$ which were not selected are imputed with the
corresponding observed coordinates of $y$. The conditions regarding
the deterministic component in \eqref{stConditions} hold as
$s_t = \nabla l(\mu | M)$, is the gradient itself.
 In Theorem \ref{thm:nmeans} we assumed that we are able to take independent draws from the truncated multivariate normal distribution, meaning that
$$
\E\left[ w_t \,\middle|\, \mathcal{F}_t \right] = \E\left[y^t - \nabla l(\mu | M)\right] = 0.
$$
In practice, we should make sure that we run the Markov chain for a sufficiently large number of iterations between gradient updates in order for \eqref{wtExp} to hold in good approximation.

The remaining issue is to bound the variance of $w_t$.  The first step
is finding an upper bound for the variance of $w_t$ as a function of
$\mu$. In the following, we denote by $f(y)$ the unconditional density
of $y$, by $f(y_j)$ the marginal (unconditional) density of $y_j$ and
by $f(y_{-j} |y_j)$ the conditional distribution of $y_{-j}$ given
$y_j$.   Since the mean minimizes an expected squared deviation we
have
\begin{align}
\E\left[\left(y_j - \E(y_j | M)\right)^2\,\middle|\, M \right]
  \nonumber
&\le \E\left[\left(y_j - \mu_j\right)^2\,\middle|\, M \right]
  \nonumber
\\ \nonumber
&=\int \left(y_j - \mu_j\right)^{2} f(y | M)\,dy 
\\ \nonumber
&=\int_{M} \left(y_j - \mu_j\right)^{2} \frac{f(y)}{P(M)}\, dy .
\end{align}
Let $C(y_j) = \int_{M} f(y_{-j}|y_j)\, dy_{-j}$, which satisfies
$0\leq C(y_j) \leq 1$.  Then
\begin{align}
\int_{M} \left(y_j - \mu_j\right)^{2} \frac{f(y)}{P(M)}\, dy
&=\int_{M} \left(y_j - \mu_j\right)^{2} \frac{C(y_j)}{P(M)}
  f(y_j)\, dy_j 
 \nonumber\\
\nonumber
&\leq \int_{M} \left(y_j - \mu_j\right)^{2} \frac{1}{P(M)} f(y_j) \,dy_j \\
&\leq \int_{\mathbb{R}} \left(y_j - \mu_j\right)^{2} \frac{1}{P(M)}
f(y_j) \,dy_j \label{sigBound} = \frac{\sigma^{2}_j}{P(M)}.
\end{align}

The next step in bounding the variance of $w_t$ is bounding $P(M)$
from below. The difficulty with finding a lower bound $P(M)$ is that
one may make it arbitrarily small by varying the coordinates of $\mu$
for the non-selected coordinates. This is the motivation behind
setting them to the observed values and only estimating the selected
coordinates, resulting in the Z-estimator described in
\eqref{mvtzest}.

Assume without loss of generality that the first $k$ coordinates of
$\mu$ were not selected and that the last $p - k + 1$ were selected.
We write
\begin{align*}
P(M) = \int_M f(y) dy = \int _{M} f(y_1 | y_2,\dots,y_p)\times
  \dots\times f(y_p) \, dy.
\end{align*}
We begin with the integration with respect to $y_1$:
$$
\int_{M} f(y_1 |y_2,\dots,y_p) \, dy_1= 1 - \Phi(u_1; \mu_{1,-1}, \sigma^{2}_{1, -1}) + \Phi(l_1 ; \mu_{1,-1}, \sigma^{2}_{1, -1}).
$$
Now, denote by $m_j = (l_j + u_j) / 2$ the mid-point between $l_j$ and
$u_j$. We have 
\begin{multline*}
1 - \Phi(u_1; \mu_{1,-1}, \sigma^{2}_{1, -1}) + \Phi(l_1 ; \mu_{1,-1}, \sigma^{2}_{1, -1}) 
\;\geq\; 1 - \Phi(u_1; m_1, \sigma^{2}_{1, -1}) + \Phi(l_1 ; m_1,
\sigma^{2}_{1, -1}) 
\\
\;\geq\;  \Phi(l_1 ; m_1, \sigma^{2}_{1, -1}) 
\;\geq\;  \Phi(l_1 ; u_1, \sigma^{2}_{1, -1}).
\end{multline*}
We can apply a similar lower bound to all selected coordinates to obtain:
\begin{align}
  \nonumber
P(M) &\geq \prod_{j\in M}\Phi(l_j; u_j, \sigma^{2}_{j, -j}) 
\int_{M} f(y_{p-k+1}|y_{p-k+2},\dots,y_p)\times\dots\times
f(y_p)\,dy_{p-k+1}\dots dy_p\\
\label{probBound}
&=P\left(j\notin S(y) \,\forall j\notin M\right) \prod_{j\in M}\Phi(l_j; u_j, \sigma^{2}_{j, -j}).
\end{align}
Taking \eqref{sigBound} and \eqref{probBound} together, we obtain the desired bound:
$$
\Var(y_j) \leq \frac{\tr(\bSigma)}{P(\bigcap_{j\notin M} \{j\notin M\}) \prod_{j\in M}\Phi(l_j; u_j, \sigma^{2}_{j, -j})}.
$$
\qed

The proof of Theorem \ref{LassoSGDconvergence} follows in a similar fashion.

\subsection{Proof of Lemma \ref{delayedRejection}}
The proposal vectors defined in the lemma are given by:
\begin{flalign*}
&r_{1}^{\rightarrow} =\left(\eta_{1}^{t},\dots,\eta^{t}_{j-1},r_j,\eta^{t-1}_{j+1},\dots,\eta^{t-1}_{p}, \xi^t \right), \quad
r_{2}^{\rightarrow} = \left(r_1,\dots,r_{j-1},r_j,r_{j+1},\dots,r_p, \xi^t \right), 
\\
&r_1^{\leftarrow} =  \left(r_1,\dots,r_{j-1},\eta^{t-1}_j,r_{j+1},\dots,r_p, \xi^t \right), \quad
r_2^{\leftarrow} = \left(\eta_{1}^{t},\dots,\eta^{t}_{j-1},\eta_{j}^{t-1},\eta_{j+1}^{t-1},\dots,\eta^{t-1}_{p}, \xi^t \right).
\end{flalign*}
The proposed algorithm for sampling $\eta | M, \xi$ is a two-step Delayed Rejection
Metropolis-Hastings sampler.  In our case the first step is to propose
a sample from the full conditional distribution of $\eta_j$ given
$\eta_{-j}$.  We denote the first proposal by $r_{1}^{\rightarrow}$.  
Note that at this stage only the $j$th coordinate has been
changed. The acceptance probability for this step is given by:
$$
\alpha(r_{2}^{\leftarrow}, r_{1}^{\rightarrow}) = 
\frac{f(r_{1,j}^{\rightarrow}|r_{1,-j}^{\rightarrow})}
{f(r_{2,j}^{\leftarrow}|r_{2,-j}^{\leftarrow})}
\frac{f(r_{2,j}^{\leftarrow}|r_{2,-j}^{\leftarrow})}
{f(r_{1,j}^{\rightarrow}|r_{1,-j}^{\rightarrow})} I\{S_n(X, r^{\rightarrow}_{1}) = M\} 
= I\{S_n(\bX, r^{\rightarrow}_{1}) = M\}.
$$
That is, the acceptance probability of the first proposal is either $1$ or $0$ depending on whether the proposal satisfies conditions \eqref{oneCondition} and \eqref{zeroCondition}. 

If the first proposal is not accepted and \eqref{zeroCondition} is satsifeid, then we make a second proposal $r^{\rightarrow}_2$. The acceptance probability for the second proposal as defined by \citet{mira2001metropolis} is given by:
$$
\alpha(r_2^{\leftarrow},r_{1}^{\leftarrow},r_2^{\leftarrow}) = 
\frac{f(r^{\rightarrow}_2) q_1(r^{\rightarrow}_2, r^{\leftarrow}_1) q_2(r^{\rightarrow}_2, r^{\leftarrow}_1, r^{\leftarrow}_2)\left(1 - \alpha(r_{2}^{\rightarrow}, r_{1}^{\leftarrow})\right)}
{f(r^{\leftarrow}_2) q_1(r^{\leftarrow}_2, r^{\rightarrow}_1) q_2(r^{\leftarrow}_2, r^{\rightarrow}_1, r^{\rightarrow}_2)\left(1 - \alpha(r_{2}^{\leftarrow}, r_{1}^{\rightarrow})\right)},
$$
where $q_1(x,y)$ is the density of the first proposal and $q_2(x,z,y)$ is the density of the second proposal.  We only make a second proposal if $\alpha(r_{2}^{\leftarrow}, r_{1}^{\rightarrow}) =0$ and therefore the ratio is always zero if $r_1^{\leftarrow}$ is a legal value. If both $r^{\leftarrow}_1$ and $r^{\rightarrow}_1$ are illegal then $\alpha(r_2^{\leftarrow},r_{1}^{\leftarrow},r_2^{\leftarrow}) $ is non-zero and the proposal densities are given by:
\begin{align*}
&q_1(x,y) = f\left(y_j \middle| \{y_j < l_j\} \cup \{u_j < y_j\},x_{-j}\right), \\
&q_2(x,z,y) = \prod_{k\neq j} \frac{\varphi(y_k ; x_k, \sigma^{2}_{k,-k})}{P(y_k \in (a_k, b_k) ; x_k, \sigma^{2}_{k,-k})}.
\end{align*}
Put together, we get:
$$
q(x, y) := q_1(x,y)q_2(x,z,y),
$$  
which yields the desired result.
\qed

\subsection{Proof of Theorem \ref{thm:naiveconsistency}}\label{naiveConsistency}
Under the assumptions of Theorem \ref{thm:naiveconsistency} we show
that the unadjusted MLE is consistent even in the presence of model selection, in the sense that:
$$
\lim_{n\rightarrow\infty} P(\| \hat\theta^M_n - \theta^M_0\|_\infty \geq \varepsilon | M)= 0.
$$
We prove this result by showing that it holds for a model $M\in \mathcal{M}$ that satisfies the conditions of the theorem. Assume without loss of generality that $\theta^M \in \Theta^M \subseteq \mathbb{R}^{p}$. In the following we will use the shorthand $I_n(M) = I_{\{ S_n(y) = M\}}$. The results follows from the fact that as long as the probability of model selection can be bounded from below, then the selection thresholds cannot be too far a way from the true parameters.

\begin{align*}
&\lim_{n\rightarrow\infty} P(\| \hat\theta^M_n - \theta^M_0\|_\infty \geq \varepsilon | M)   \\
&= \lim_{n\rightarrow\infty} \frac{P_n(M | \{\| \hat\theta^M_n - \theta^M_0\|_\infty \geq \varepsilon\}) P(\| \hat\theta^M_n - \theta^M_0\|_1 \geq \varepsilon)}{P_n(M)}  \\
&\leq \lim_{n\rightarrow\infty} \frac{P(\| \hat\theta^M_n - \theta^M_0\|_\infty \geq\varepsilon)}{P_n(M)} \\
&=  \lim_{n\rightarrow\infty} \frac{P\left(\bigcup_{j=1}^{p} \{|\hat\theta^M_{nj} - \theta^M_{0j}| \geq \varepsilon \} \right)}{P_n(M)}  \\ 
&\leq   \lim_{n\rightarrow\infty}  \sum_{j=1}^{p} \frac{P(|\hat\theta^M_{nj} - \theta^M_{0j}| \geq\varepsilon)}{P_n(M)}  \\
&=   \lim_{n\rightarrow\infty}  \sum_{j=1}^{p} \frac{P( |\sqrt{n}(\hat\theta^M_{nj} - \theta^M_{0j})| \geq \sqrt{n}\varepsilon)}{P_n(M)} \\
&\leq^{(*)} \lim_{n\rightarrow\infty}  \sum_{j=1}^{p} \frac{2e^{\frac{-n\varepsilon^{2}}{2\sigma_{Mj} }}}{P_n(M)}=^{(**)} 0 ,
\end{align*}
where $\sigma^{2}_{Mj}$ is the $j$th diagonal element of $\bSigma^M$
and $(*)$ holds by subgaussian concentration. The equality $(**)$ holds by our assumption regarding the rate at which $P_n(M)$ is allowed to tend to zero. \qed

\subsection{Proof of Theorem \ref{thm:consistency}} 

Before we prove the theorem, we first state and and prove a couple of Lemmas that will come in handy in the proof of Theorem \ref{thm:consistency}. Lemma \ref{conditionalinnaive} to follow states that the conditional MLE is consistent for $\theta^M_0$ even when used in the non-conditional setting (when the model to be estimated is pre-determined).

\begin{lemma}\label{conditionalinnaive}
Set a family of distributions $M$ and assume that no data-driven model selection has been performed.  Then under the conditions of Theorem \ref{thm:consistency} the conditional MLE is consistent for $\theta^M_0$, that is,
$$
P(\|\hat\theta^M_n - \theta^M_0\|_\infty > \varepsilon) \rightarrow 0.
$$
\end{lemma}

\proof
Consider once again the conditional MLE
\begin{align*}
\hat{\theta}^{M}_n &= \arg\max_{\theta^M} G_n^M  
= \arg\max_{\theta^M} \frac{1}{n} \sum_{i=1}^{n}\left[ \ell_{\theta^M}(y_i)-\frac{1}{n}\log P_{n,\theta^M}(M) \right] \\
& := \bar{\ell}_n(\theta^M) - \frac{1}{n}\log P_{n,\theta^M}(M).
\end{align*}
where $ \ell_{\theta^M}(y_i)$ is the unconditional log-likelihood of $y_i$. We are evaluating the properties of the conditional estimator in the unconditional setting where $M$ is designated for inference before the data are observed. In this setting, the conditional MLE can be considered an M-estimator obtained from performing inference under a misspecified likelihood. 

We now show that $\hat\theta_n^M$ is consistent for the
$\theta^M_0$.  We have
$$
\sup_{\theta^M}{G}_n^{M}(\theta^M) \geq {G}_n^{M}(\theta_0^M),
$$
which implies that
\begin{equation}\label{firstProbDiff}
 \bar{\ell}_n(\hat{\theta}^{M}_n) \geq
\bar{\ell}_n(\theta_0^{M})
- \frac{1}{n}
\log P_{n,\theta^M_0}(M) + \frac{1}{n} \log P_{n,\hat{\theta}^M_n}(M).
\end{equation}
Equation \eqref{firstProbDiff} together with assumption \eqref{exp-inf-prob} gives
\begin{equation}\label{op1eq}
\bar{\ell}_n(\tilde{\theta}^M_n) \geq \bar{\ell}_n(\theta_0^M) - o(1).
\end{equation}
Thus, the conditions for consistency as given by
\citet{van2000asymptotic} (Theorem 5.14 p.~48) are satisfied. The implication of \eqref{op1eq} is that in the unconditional setting the conditional M-estimator is a consistent estimator. 
\qed

Next, we show that the difference between the conditional expectation of the sufficient statistic $\bar{T}_n$ converges to the unconditional expectation. This result will assist us later in proving a law-of-large number type statement for $\bar{T}_n$ under the conditional distribution.  

\begin{lemma}\label{expectationLemma}
Under the assumptions of Theorem \ref{thm:consistency}, for all $\delta < 1/2$,
$$
 n^{\delta} \| \E_{\theta^M_0}(\bar{T}_n) -  \E_{\theta^M_0}(\bar{T}_n|M) \| \rightarrow 0 .
$$
\end{lemma}

\proof According to Lemma \ref{scoreLemma}, if $y_i \sim f_{\theta^M}$ with
$f_{\theta^M}$ an exponential family distribution and $P_{n,\theta^M}(M|\bar{T}_n) \in \{0, 1\}$
then the first derivative of the conditional log-likelihood is 
$$
\frac{\partial}{\partial \theta^M} G_n(\theta^M) 
= \frac{1}{n} \sum_{i=1}^{n} T(y_i) - \E_{\theta^M}(T(y_i) | M)
:= \bar{T}_n - \E_{\theta^M}(\bar{T}_n |M).
$$
At the maximizer of $G_n(\theta^M)$, for any $\delta < 1/2$, we have:
$$
n^{\delta} \left[ \bar{T}_n - \E_{\hat\theta^M_n}(\bar{T}_n |M) \right] = 0 ,
$$
which implies that
$$
n^{\delta}(\bar{T}_n - \E_{\theta^M_0}(\bar{T}_n) )
+ n^{\delta}(\E_{\theta^M_0}(\bar{T}_n) -  \E_{\hat\theta^M_n}(\bar{T}_n |M)) = 0 .
$$
Since $n^{\delta}(\bar{T}_n - \E_{\theta^M_0}(\bar{T}_n) )=o_p(1)$ by law of large numbers, we obtain that
$$
n^{\delta}(\E_{\theta^M_0}(\bar{T}) -  \E_{\hat\theta^M_n}(\bar{T} |M)) = o_p(1).
$$

Finally in order to prove the desired results we must show that
$$
E_{\theta^M_0}(\bar{T} |M) - E_{\hat\theta^M_n}(\bar{T} |M) \rightarrow 0.
$$
It is clear that since $\theta^M_n \rightarrow\theta^M_0$, a fixed continuous function of $\hat\theta^M_n$ will converge as the sample size grows. However,  $E_{\hat\theta^M_n}(\bar{T} |M)$ is a function of both $\hat \theta^M_n$ and $n$, and we must make sure that it does not vary too much with $n$ in order for the desired convergence to hold. Define $t = a^{T}\bar{T}_n$. By  assumption \eqref{newSup} we have that for some sufficiently large $n$:
$$
\sup_{\theta^M: \|\theta^M - \theta^M_0\| < \frac{1}{\sqrt{n}}} 
|\E_{\theta^M_0}(t |M) - \E_{\theta^M}(t|M)| 
\leq \sup_{\theta^M: \|\theta^M - \theta^M_0\| < \frac{1}{\sqrt{n}}} \frac{\Var_{\theta^M}(t)}{P_{n,\theta^M}(M)} \frac{1}{\sqrt{n}}.
$$ 
Because $y$ is of an exponential distribution and $t$ is an average we can bound the unconditional variance in the neighborhood of $\theta^M_0$. For a sufficiently small $\varepsilon > 0$ there exists a constant $C > 0$ such that, 
$$
\sup_{\theta^M: \|\theta^M - \theta^M_0\| \leq \varepsilon } \Var_{\theta^M}(t) < \frac{C}{n}
$$ 
because $\Var_{\theta^M}(t)$ is a continuous function and the supremum is taken over a compact set. Thus, by the $\sqrt{n}$ consistency of $\hat\theta^M_n$ for $\theta^M_0$, the difference satisfies $n^{\delta} |\E_{\theta^M_0}(t |M) - \E_{\hat\theta_n^M}(t|M)| = o(1)$ for any vector $a$ as well as for $\bar{T}_n$ itself and the claim follows. 
\qed

\vspace{0.1 cm}

We are now ready to prove Theorem \ref{thm:consistency}. The first step in the proof is showing that $\bar{T}_n$ converges in probability conditionally on $M$. This result is a simple consequence of Markov's inequality and our assumption that $P_{n}(M)^{-1} = o(n)$. Set an arbitrary vector $a \in \mathbb{R}^{p}$ and define $t = a^{T} \bar{T}_n$. By Markov's inequality,
\begin{equation}\label{eq:markov}
P_n(|t - \E_n(t|M)| > \varepsilon|M) 
\leq \frac{\Var_n(t|M)}{\varepsilon^{2}} 
\leq \frac{O(n^{-1})}{\varepsilon ^{2} P_n(M)} = o(1).
\end{equation}
To see why \eqref{eq:markov} holds, write:
$$
\Var_n(t|M) = \int \frac{(t - \E(t))^{2}}{P_n(M)} I\{S_n(\bar{T}_n) = M\}f(t)d(t) \; -\; \left[\E(t) - \E_n(t|M)\right]^{2}
$$$$
\leq \frac{a^{T} \Var( T(y_i)) a}{nP_n(M)}.
$$
By the fact that \eqref{eq:markov} holds for any arbitrary vector $a$, together with Lemma \ref{expectationLemma}, we can determine that conditionally on $M$, $\bar{T}_n \rightarrow_p \E(\bar{T}_n)$. 

By our assumption that  the log-likelihood $l_{\theta^M}(y)$ is a continuous mapping of $T(y)$, assumption \eqref{exp-inf-prob} and Lemma \ref{expectationLemma}, conditionally on the selection of $M$ we have:
$$
\frac{1}{n} \sum_{i=1}^{n} \ell_{\theta^M}(y_i) - \frac{1}{n} \log P_{n,\theta^M}(M) \rightarrow_p  \E[\ell_{\theta^M}(y_i)].
$$
The rest of the proof follows in a similar manner to the proof of Lemma \ref{conditionalinnaive} where the law of large numbers in the proof of Theorem 5.14 in \citet{van2000asymptotic} is replaced by \eqref{eq:markov} and our assumption that $\bar{\ell}_n(\theta^M)$ is a continuous function of $\bar{T}_n$. 
\qed

\subsection{Proof of Lemma \ref{slowLemma}}
In the context of this proof we use the following notation:
$$
A_0(M,s) :=   \left\{ l_o(M, s) \leq \bA_0(M, s) y < u_0 (M, s)\right\},
$$$$
A_1(M,s) := \left\{  \bA_1(M, s) y < u_1 (M, s)\right\} .
$$
For ease of exposition, we make a simplifying assumption that 
$$
\lim_{n\rightarrow\infty} \frac{\lambda_n}{n^{\frac{1}{2}}} = \lambda^*.
$$

We begin by bounding the probability of not selecting the null-set. By our assumption that $n^{-1}\bX^{T}\bX$ converges, we have that the thresholds $l_0(M,s)$ and $u_0(M,s)$ also convergence for all candidate models and sign permutations. Furthermore, by our assumption regrading the rate in which $\lambda_n$ grows and the expectation of $\bA_0(M)y$, 
$$
\bA_0(M) y \rightarrow^D N(0, \bSigma(\bA_0)), 
$$
where,
$$
\bSigma(\bA_0) = \lim_{n\rightarrow\infty} \frac{\sigma^{2}}{\lambda_n^{2}} 
\bX^{T}_{-M}(I - \bX_M(\bX^{T}_M \bX_M)^{-1} \bX^{T}_M)\bX_{-M}.
$$
Thus, 
$$
\lim_{n\rightarrow\infty} P_n(\bA_0(M,s)) = c_0(M,s) > 0, \;\;\forall M, s.
$$
Since the probability of $\bA_0(M,s)$ can be bounded in a uniform manner, we can set
$$
c_0(M) := \min_{s} c_0(M,s), 
$$
and obtain a lower bound for the probability of selecting $M$ by bounding
$$
P_n(M) \geq c_0(M) P_n\left( \cup_{s}  \bA_{1}(M,s) \right) := c_0(M) P_n\left(\bA_1(M) \right).
$$

We bound $P_n(\bA_1(M))$ next. Recall that the threshold a regression coefficient must cross is
given by
$$
u_1(M,s) = -\lambda_n \; \text{diag}(s) (\bX^{T}_M \bX_M)^{-1} s.
$$
This threshold is a bit unwieldy, as it depends on the signs of the active set and an exact realization of $\bX_M$. Since we are interested in asymptotic behavior of random quantities, it will be sufficient to work with the limiting value of the threshold:
$$
u^*_1(M,s) = \lim_{n\rightarrow\infty} \sqrt{n} u_1(M,s) = -\lambda^* \; \text{diag}(s) \bSigma_M^{-1} s,
$$
Now, in order to eliminate the dependence on the signs of the active set define:
$$
u^*_1(M) := \sup_{s}\sup_{j} \left|\left(\lambda^* \; \text{diag}(s) \bSigma_M^{-1} s\right)_j\right|,
$$
and define an event:
$$
\tilde{\bA}_1 := \{\sqrt{n}|\eta_j| > u^*_1(M), \;\; \forall j \in M\}.
$$
In $\tilde A_1$ we replaced all coordinate thresholds with the largest threshold, and so it is clear that:
$$
\lim\sup_{n\rightarrow\infty}\frac{P_{n,\beta^M}(\tilde{A}_1)}{P_{n,\beta^M}(A_1)} \leq 1.
$$
Furthermore, we have the lower bound
\begin{equation}\label{AtildeBound}
P_{n,\beta^M} (\tilde{A}_1) \geq \prod_{j \in M} \left( \Phi(-u^*(M) ; 0, \sigma^{2}_{j, -j}) + 1 -  \Phi(u^*(M) ; 0, \sigma^{2}_{j, -j})\right)
, \;\;\;\ \forall\beta^M \in \mathbb{R}^{|M|},
\end{equation}
where $\sigma^{2}_{j, -j} := Var(\sqrt{n}\eta_j | \eta_{-j})$. See the proof of Theorem \ref{thm:nmeans} for details on how this bound is derived. The rest follows by our normality assumption and the fact that \eqref{AtildeBound} holds for all $\beta^M$ including $\beta^M_0$.
\qed

\subsection{Proof of Lemma \ref{fastLemma}}
We begin by treating the probability of satisfying the conditions for not selecting the variables not in the model. Using the same notations as in the proof of Lemma \ref{slowLemma}, the following limit holds:
$$
\bSigma(A_0) \rightarrow 0, 
$$
and consequently, by assumption \eqref{irrep}:
$$
\lim_{n\rightarrow\infty}P_n(\bA_0(M,s)) = 1, \;\;\forall s.
$$

Next, we treat the probabilities of satisfying the conditions for selecting the variables included in the model. As before, we make a simplifying assumption that there exists a constant $0 < \delta <0.5$ such that:
$$
\frac{\lambda_n}{n^{0.5 + \delta}} = \lambda^*, 
$$
In the fast scaling case, a lower bound on $P_{n,\beta^M}(\tilde{A}_1)$ no longer exists because the threshold $u^*(M)$ grows with the sample size. However, we can show that a satisfactory bound exists at $\beta^M_0$. Since in this setting $\lambda_n$ grows faster than $\sqrt{n}$, we redefine the limit of the selection threshold:
$$
u^*_1(M,s) = \lim_{n\rightarrow\infty}\frac{\sqrt{n}}{n^{\delta}} u_1(M,s) = -\lambda^* \; \text{diag}(s) \bSigma_M^{-1} s.
$$
We can redefine $u^*_1(M)$ in an analogous manner. Now, we rewrite the bound \eqref{AtildeBound} at the point $\beta^M = \beta^M_0$ and with $u^*_1(M)$ properly scaled as
$$
P_{n,\beta^M_0} (\tilde{A}_1) \geq \prod_{j \in M} \left( \Phi(- u^*(M) n^\delta; \sqrt{n}\beta^M_0, \sigma^{2}_{j, -j}) + 1 -  \Phi( u^*(M) n^\delta; \sqrt{n}\beta^M_0, \sigma^{2}_{j, -j})\right)
$$
With no loss of generality assume that $\beta^M_0 < 0$ to obtain the desired bound:
$$
\lim_{n\rightarrow \infty} P_{n,\beta^M_0} (\tilde{A}_1) \geq 
\lim_{n\rightarrow\infty} \prod_{j \in M} \Phi\left(- u^*(M) n^\delta; \sqrt{n}\beta^M_0, \sigma^{2}_{j, -j}\right) = 1,
$$
where the limit holds because $\delta < 0.5$. A similar result holds in a small neighborhood $U$ of $\beta^M_0$ because the probability of selection is continuous in $\beta^M$.

In order to bound the infimum of $P_{n,\beta^M}(A_1)$, we again start from \eqref{AtildeBound} to get:
\begin{align} \nonumber
P_{n, \beta^M} (\tilde{A}_1) &\geq 
\prod_{j \in M} \left( \Phi\left(- u^*(M) n^{\delta};0, \sigma^{2}_{j, -j}\right) + 1 -  \Phi\left( u^*(M) n^{\delta}; 0, \sigma^{2}_{j, -j}\right)\right) 
\\ \nonumber
&\geq \prod_{j \in M}\Phi\left(- u^*(M) n^{\delta}; 0, \sigma^{2}_{j, -j}\right) 
\\
\label{normUpperbound}
&\geq C \left(\frac{n^{\delta} u^*_1(M) / \sigma_{j,-j}}{1 +  u^*_1(M)^2 n^{2\delta} / \sigma^{2}_{j, -j}}\right)^{|M|} 
\prod_{j \in M} e^{-\frac{ u_1^*(M)^{2} n^{2\delta}}{2\sigma^{2}_{j, -j}}} \\ \nonumber
&= O\left(\frac{e^{-n^{2\delta}}}{n^{\delta|M| / 2}} \right).
\end{align}
The lemma follows by our assumption that $\delta < 0.5$. In \eqref{normUpperbound} we used the inequality:
$$
\Phi(t;0, \sigma^{2}) \geq C \frac{t/\sigma}{1 + t^{2} / \sigma^{2}} e^{-\frac{t^{2}}{2\sigma^{2}}}.
$$ 
\qed

\section{Numerical examples for the Lasso MLE}\label{LassoExampleAppendix}

\begin{figure}[t] 
\includegraphics[width=\maxwidth]{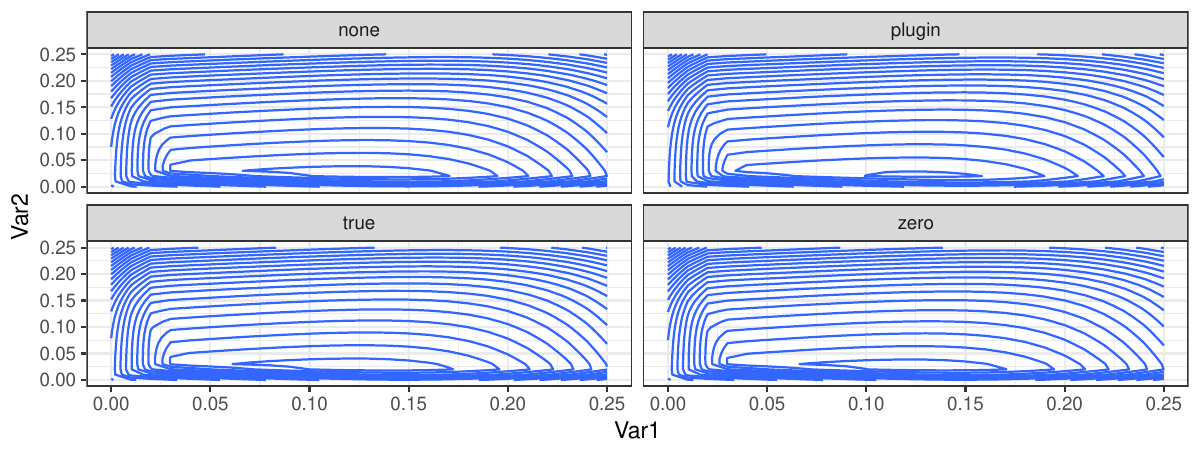} 
\caption[Contours of the conditional Lasso likelihood - A]{Contour plots for the first numerical experiment described in Appendix \ref{LassoExampleAppendix}. The contour plots describe the log-likelihood of a model selected by the Lasso as a function of the values of the regression coefficients where the probability of not selecting the inactive set was computed in four different ways as described in the text.} 
\label{LassoContours1} \end{figure}

\begin{figure}[t] 
\includegraphics[width=\maxwidth]{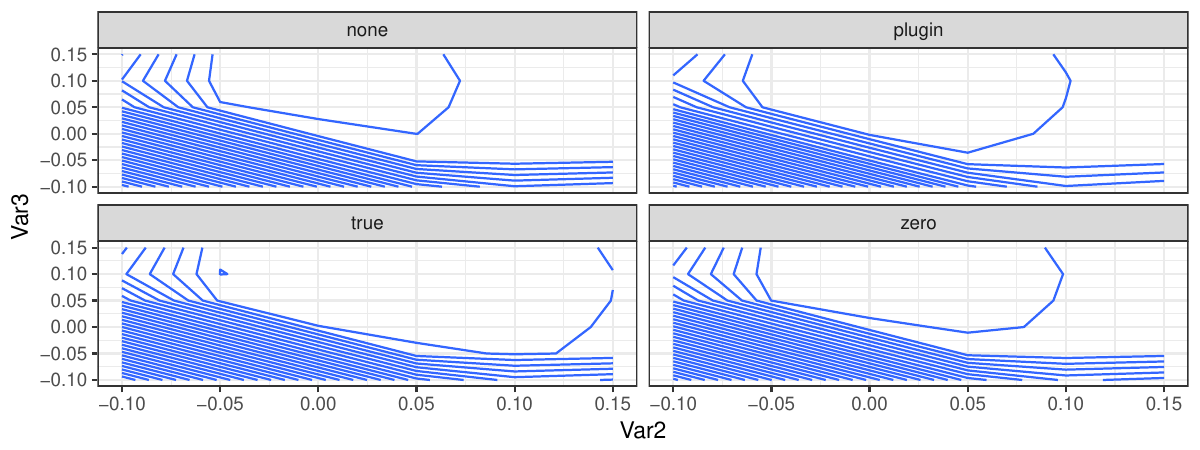} 
\caption[Contours of the conditional Lasso likelihood - B]{Contour plots for the second numerical experiment described in Appendix \ref{LassoExampleAppendix}. The contour plots describe the log-likelihood of a model selected by the Lasso as a function of the values of the regression coefficients where the probability of not selecting the inactive set was computed in four different ways as described in the text.} 
\label{LassoContours2} \end{figure}

In Section \ref{sec:lassoSelect} we discuss the conditions that must hold in order for a specific model to be selected by the Lasso and propose to estimate the mean vector $\bA_0(M)E(y)$ by $0$. Here, we propose some alternatives and seek to demonstrate that the proposed method is a reasonable one. 

We generate data using the same process as described in Example \ref{LassoExample} with parameter values $\rho = 0.5$, $n = p = 100$, $k = 3$ and $\snr = 0.5$. We selected a model with two active parameters of positive sign with observed values of $0.17$ and $0.13$. In order to  compute the conditional log-likelihood for this example we must decide on appropriate estimates for $E(\bA_0(M)y)$. We present results for three options. The first is to use the observed value, $\bA_0y$ as an estimate for its expectation, we term this method `plug-in'. The second is to work under the assumption that $E(\bA_0y) \approx 0$, estimating the expectation with a vector of zeros, we term this method `zero'. A third option is to simply assume that $P(l < A_0y <u )\approx 1$ for all signs sets, we term this method `none'. Finally, we also compute the likelihood under the truth, setting $E(\bA_0 y) = \bA_0 E(y)$.

We draw the contour plots for the two-dimensional log-likelihoods as a function of the selected regression coefficients in Figure \ref{LassoContours1}. While the contour plots are visually similar, the values of the log-likelihoods differ slightly. For the `none' and `zero' methods the log-likelihood was maximized at $0.14, 0.02$ at a log-likelihood value of $14.2$. This is similar to the log-likelihood computed under the true expectation, where the maximum was also obtained at $0.14, 0.02$ and at a slightly different value of $14.3$. Finally, for the plug-in method the maximum was obtained at $0.13, 0.02$ with a value of $16.9$. Thus, for this example, the maximum likelihood estimates computed using the different imputation methods yielded results that are essentially equivalent. In this example the true probability of $P(l_0 < \bA_0y < u_0)$ was close to $1$ for all sign permutations.

 In a second example we generate data using parameter values $\rho = 0.8$, $n = 100$, $p = 500$, $k = 5$ and $\snr = 0.2$. Here we selected a model with four variables where the observed refitted regression coefficients estimates were $0.13, 0.17, 0.21$ and $0.15$. For all estimation methods the maximum of the log-likelihood was obtained at  approximately $0, -0.05, 0.1, 0$. The values of the log-likelihood function at its maximum was $15.9$ when no imputation was used, $19.9$ for plugin imputation, $16.1$ for the zero imputation and $16.7$ when the true parameter value was used to compute the log-likelihood. The contour of the log-likelihood function are plotted in Figure \ref{LassoContours2} for the second and third variables, keeping the values of the first and last coefficients fixed at zero.

\section{Description of algorithms}
\begin{algorithm}\label{mvtSGA}
\IncMargin{1em}
\SetKwInOut{Input}{input}
\SetKwInOut{Output}{output}
\SetKwInOut{Initialize}{initialization}

\Input{$y, l, u\in\mathbb{R}^{n}$, $\Sigma^{-1} \in \mathbb{R}^{p\times p}$.}
\Output{$\hat\mu\in\mathbb{R}^{p}$.}
\BlankLine
\Initialize{$y^{0} \leftarrow y$, $\mu^0 \leftarrow y$.}
\BlankLine
\For{$i \in 1:I$} {
	Set $z^0 \leftarrow y^{i-1}$\;
	\For{$t \in 1:T$} {
		\For{$j \in 1:p$} {
		Sample $z^{t}_j \sim f_{\mu^i}(z_j | M, z_1^{t},\dots,z_{j-1}^{t},z_{j+1}^{t-1},\dots,z_{p}^{t-1})$\;
		}
	}
	Set $y^i \leftarrow z^{T}$\;
	\For{$j \in M$} {
		$\mu^{i}_j \leftarrow \mu^{i-1}_j + \gamma^{i}\Sigma^{-1}_{j,.}(y - y^{i})$\;	
	}		
}
\Return{$\mu^I$\;}
\caption{Stochastic ascent algorithm for the normal means problem.}
\end{algorithm}\DecMargin{1em}

%%%% MH ALGORITHM

\begin{algorithm}\label{LassoMH}
\SetKwInOut{Input}{input}
\SetKwInOut{Output}{output}
\SetKwInOut{Initialize}{initialization}
\Input{$\eta\in\mathbb{R}^{|M|}$, $\lambda \in \mathbb{R^+}$, $X \in
  \mathbb{R}^{n\times p}$, $\sigma^{2} \in \mathbb{R}^{+}$.}

\Output{A sample point $\eta$.}
\BlankLine

\For{$t \in 1:T$} {
	Sample $\xi^{t} \sim f(\xi|M, \eta)$ \;
	\For{$j \in 1:p$} {
		Set $r^{\rightarrow} \leftarrow \eta$ \;
		Sample $r_j^{\rightarrow} \sim f\left(\eta_j \middle| \{\eta_j < l_j\} \cup\{ \eta_j > u_j\},\eta_{-j}\right)$\;
		\If{$l_0(M, \sign(r^{\rightarrow})) < \xi^{t} < u_0(M, \sign(r^{\rightarrow}))$}{
		\If{$r^{\rightarrow}$ is in the set from \eqref{oneCondition}} {
			Set $\eta \leftarrow r^{\rightarrow}$ \;
		}
		\Else{
			\For{$k\neq j$} {
				Sample $r_{k}^{\rightarrow} \sim \text{TN}(a_k,b_k, \eta_k, \sigma^{2}_{k, -k})$ \;
			}
		Set $r^{\leftarrow} \leftarrow r^{\rightarrow}$ \;
		Set $r^{\leftarrow}_j \leftarrow \eta_j$ \;
		\If{$r^{\leftarrow}$ is not in the set from \eqref{oneCondition}} {
		Compute $p^t_j$ as in \eqref{MHtransition} \;
		Sample $U \sim \Unif(0,1)$\;
		\If{$U < p^t_j$} {
			Set $\eta \leftarrow r^{\rightarrow}$\;
		}
	}
	}}	
}
}
\Return{$\eta$ \;}
\caption{Sampler for the post-selection distribution under selection
  by Lasso.}
\end{algorithm}
%%%%%%%%%%%%%

%%%% Lasso SGD ALGORITHM

\begin{algorithm}\label{LassoSGD}
\IncMargin{1em}
\SetKwInOut{Input}{input}
\SetKwInOut{Output}{output}
\SetKwInOut{Initialize}{initialization}

\Input{$I \in \mathbb{N}$, $\lambda, \sigma^{2} \in \mathbb{R}^+$, $X \in \mathbb{R}^{n\times p}$, $y \in \mathbb{R}^{n}$.}
\Output{$\hat\beta \in \mathbb{R}$.}
\BlankLine
\Initialize{Set $\hat\beta^0, \eta^0 \leftarrow (X^{T}_MX_M)^{-1} X_M^{T} y$.}
\BlankLine

\For{$i \in 1:I$} {
	Sample $\eta^i$ using Algorithm \ref{LassoMH} \;
	Set $\hat\beta^{i} \leftarrow \hat\beta^{i-1} + \gamma_i(X_m^{T}y - (X^{T}_MX_M)\eta^i)$\;
	\For{$j \in 1:p$} {
		Set $\hat\beta_j^t \leftarrow \sign(\hat\beta^0_j) \max(0,\sign(\hat\beta^0_j)\hat\beta^i_j)$\;
	}
}
\Return{$\hat\beta^{I}$}\;
\caption{Stochastic ascent algorithm for the Lasso.}
\end{algorithm}\DecMargin{1em}

\end{document}